\def\btog{1} 

\ifx\btog\undefined
    \documentclass[twocolumn]{article}
\else
    \documentclass[      acmtog,fleqn,authorversion,nonacm,natbib=false,pdfpagelabels=false]{acmart} 
\fi

\ifx\btog\undefined
    \RequirePackage{geometry}
    \usepackage[fontsize=11pt]{scrextend} 
    \usepackage{amsmath,amsfonts,amssymb}
    \usepackage{graphicx}
    \usepackage{xcolor}
    \usepackage[colorlinks=true]{hyperref}

    \newenvironment{teaserfigure}
        {\begin{figure*}}
        {\end{figure*}}

    \usepackage{amsthm}
    
    \newtheorem{definition}{Definition}

    \newtheorem{lemma}{Lemma}
    \begingroup 
        \makeatletter
        \@for\theoremstyle:=definition,remark,plain\do{%
            \expandafter\g@addto@macro\csname th@\theoremstyle\endcsname{%
                \addtolength\thm@preskip\parskip
                }%
            }
    \endgroup
\else
    \acmSubmissionID{0}
\fi

\usepackage[shortlabels]{enumitem} 
\usepackage[short,nocomma]{optidef}
\graphicspath{{fig/}}
\setkeys{Gin}{keepaspectratio} 
\usepackage{wrapfig}

\usepackage{cleveref} 
\crefname{algocf}{alg.}{algs.} 

\usepackage[linesnumbered,ruled,vlined]{algorithm2e} 

\SetCommentSty{mycommfont}

\usepackage{thm-restate} 
\usepackage{bbold} 
\allowdisplaybreaks 
\usepackage{mathtools}

\usepackage{enumitem}
\setlist[description]{leftmargin=*}
\setlist[itemize]{leftmargin=*}
\setlist[enumerate]{leftmargin=*}

\usepackage[autostyle, english = american]{csquotes}
\MakeOuterQuote{"} 

\newcommand{\descrip}{\textit} 

\usepackage[datamodel=acmdatamodel,style=acmauthoryear,backend=biber]{biblatex}
\newcommand*\autoop{\left(}
\newcommand*\autocp{\right)}
\newcommand*\autoob{\left[}
\newcommand*\autocb{\right]}
\AtBeginDocument {%
   \mathcode`( 32768
   \mathcode`) 32768
   \mathcode`[ 32768
   \mathcode`] 32768
   \begingroup
       \lccode`\~`(
       \lowercase{%
   \endgroup
       \let~\autoop
   }\begingroup
       \lccode`\~`)
       \lowercase{%
   \endgroup
       \let~\autocp
   }\begingroup
       \lccode`\~`[
       \lowercase{%
   \endgroup
       \let~\autoob
   }\begingroup
       \lccode`\~`]
       \lowercase{%
   \endgroup
       \let~\autocb
   }}
\makeatletter
\AtBeginDocument {%
          \def\resetMathstrut@{%
           \setbox\z@\hbox{\the\textfont\symoperators\char40}%
           \ht\Mathstrutbox@\ht\z@ \dp\Mathstrutbox@\dp\z@}%
}%
\makeatother

\DeclareMathOperator{\Dif}{\mathsf{D}}

\DeclareMathOperator{\cell}{cell}
\DeclareMathOperator{\cdist}{cdist}

\ifdefined\btog
    \acmJournal{CIE}\acmVolume{0}\acmNumber{0}\acmArticle{0}\acmYear{202x}\acmMonth{0}
\fi

\begin{document}

\newcommand{\R}{\mathbb{R}} 
\newcommand{\Z}{\mathbb{Z}} 
\newcommand{\N}{\mathbb{N}} 
\DeclarePairedDelimiter\abs{\lvert}{\rvert} 
\DeclarePairedDelimiter\norm{\lVert}{\rVert} 
\DeclarePairedDelimiter{\round}\lfloor\rceil 
\DeclarePairedDelimiter\ceil{\lceil}{\rceil} 
\DeclarePairedDelimiter\floor{\lfloor}{\rfloor} 
\DeclarePairedDelimiterX\set[1]\lbrace\rbrace{\setaux#1}
 \def\setaux#1|{#1\;\delimsize\vert\;}
\newcommand{\halfpi}{\frac{\pi}{2}}
\newcommand{\dg}{^\circ}

\newcommand\inv[1]{#1^{-1}}

\newcommand{\maxImgHeight}{0.955\textheight}

\title{Cell-Constrained Particles for Incompressible Fluids} 

\author{Zohar Levi}

\ifdefined\btog
\affiliation{%
}
\renewcommand\shortauthors{Zohar Levi}
\fi

\tableofcontents

\begin{abstract}
Incompressibility is a fundamental condition in most fluid models.
Accumulation of simulation errors violates it and causes volume loss.
Past work suggested correction methods to battle it.
These methods, however, are imperfect and in some cases inadequate.
We present a method for fluid simulation that strictly enforces incompressibility based on a grid-related definition of discrete incompressibility.

We formulate a linear programming (LP) problem that bounds the number of particles that end up in each grid cell. A variant of the band method is offered for acceleration, which requires special constraints to ensure volume preservation. Further acceleration is achieved by simplifying the problem and adding a special band correction step that is formulated as a minimum-cost flow problem (MCFP).
We also address coupling with solids in our framework and demonstrate advantages over prior work.
\end{abstract}

\begin{teaserfigure}
    \centering
    \includegraphics[width=.9\textwidth]{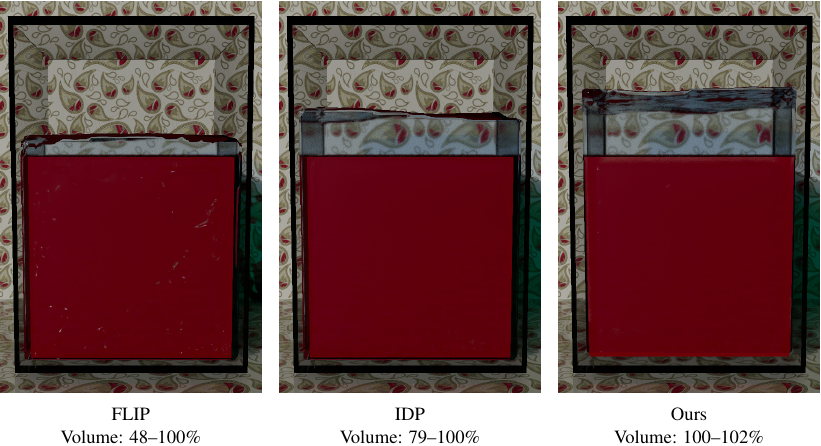}
    \caption{
    A large box is dropped into water.
    For FLIP \cite{zhu05} and IDP \cite{kugelstadt19}, particles on the bottom of the tank do not manage to clear path in time, and they are trapped inside the solid box, leading to significant volume loss.
    The volume range over all time steps is indicated below each image.
    }
    \label{fig.large_falling_obs_3d}
\end{teaserfigure}

\maketitle

\section{Introduction}
Fluids behave in a rich and complex way. Simulating them faithfully has been an active research in computer graphics to improve the realism of an animation of a large variety of materials.
Incompressibility is a fundamental condition in a fluid model, and it is expressed via a divergence-free constraint, which leads to volume conservation.
The constraint restricts instantaneous movements of particles, and it has no long term view of a fluid during a simulation.
No matter how accurate a simulation is, inaccuracies and numerical errors accumulate overtime and become pronounced, and nothing accounts for changes in the fluid's volume.

To address that, correction methods have been offered.
They range from improving particle spacing \cite{ando12} to a stricter density correction via an additional solution of a Poisson equation \cite{kugelstadt19}, or treating particles as volume parcels with prescribed volume \cite{qu22}.
While some of the methods are quite effective, they are still not perfect, and in certain cases they are inadequate.
We take a discrete approach to the problem, which restricts particles to grid cells, and show advantages over the state of the art.

We start by defining discrete incompressibility based on grid resolution (\cref{sec.discrete_inc}).
We offer a correction step to the PIC/FLIP framework \cite{bridson15fluid} that preserves this discrete incompressibility in each time step.
We limit movements of particles to specific locations in a local neighborhood of grid cells (\cref{sec.grid_move}). Incompressibility is expressed through discrete constraints that bound the number of particles that end up in each cell. These are added as hard constraints to an integer linear programming (ILP) problem, which we show can be relaxed to a linear programming (LP) problem to improve running time.
Nevertheless, the problem does not scale well, and running time can be an issue even for moderate size 3D grids.
We offer a variant of the band method \cite{ferstl16band}, tailored to our approach to preserve incompressibility (\cref{sec.band_method}).
This requires monitoring the amount of particles that go into and out of the band.
For that, we formulate an additional special band constraint to restrict particle movements near the band interface.
Two faster variants are offered.
Both first solve the LP as is (without the additional band constraint), which results in an easier problem that is faster to solve. This is followed by a second specialized step to correct the band interface and deeper.
One variant solve the LP again with a shorter one-way band constraint (\cref{sec.band_con2}).
The other variant formulates a minimum-cost flow problem (MCFP) and uses a fast algorithm based on Dijkstra's algorithm to solve it (\cref{sec.push_part}).
Overall, while solving the LP does not scale well, our fast variant of the band method has reasonable performance on moderate size grids.

Coupling with solids is a fundamental problem, and we address it within our framework that maintains incompressibility (\cref{sec.soilds}). In our evaluation, we devised scenarios (\cref{sec.scenes}) that illustrate the advantage of our method over the state of the art. One aspect that we show is that gradual correction over time, which previous works have in common, may not be adequate in certain scenarios:
There may not be an opportunity to correct the fluid after an obstacle moved, and guaranteeing incompressibility in each time step is therefore necessary.

\section{Related Work} \label{sec.related_work}
Fluid simulation was introduced to computer graphics by \citet{foster96}.
It was popularized by \citet{stam99}, who took the Eulerian view, using a grid to discretize the fluid. A semi-Lagrangian approach that simulates a particle movement was used in the advection step to ensure unconditional stability of the simulation.
\citet{kim07} simulate bubbles using the level set method. They track the volume change of each connected region and compensate for errors using divergence.

Smoothed-particle hydrodynamics (SPH) \cite{koschier22} takes the Lagrangian view, representing fluids with particles.
The approach was used in simple and intuitive methods such as \cite{muller03, macklin13}, reminiscent of the boids algorithm \cite{reynolds87}, where from local rules for a particle, based on its neighborhood, emerges a global behavior.
\citet{bender17} combine two pressure solvers, one enforcing a divergence-free
velocity field, the other satisfies a constant density condition.
\citet{band18} improve the implicit incompressible SPH (IISPH) \cite{ihmsen14} by realizing a consistent pressure gradient at boundary samples, using a different discretization of the pressure equation.

Hybrid schemes based on the particle-in-cell (PIC) and fluid-implicit-particle (FLIP) methods \cite{zhu05} use a dual view combining grid and particles for fluid representation.
The approach is similar to the material point method (MPM) \cite{jiang16}, which was used to simulate a larger variety of materials including elasto-plastic constitutive models.
Between the algorithms steps, data is transferred between the two representations, which causes a loss of information. A few methods aim at mitigating the loss. APIC \cite{jiang15apic} endows each particle with additional information in the form of a matrix, which allows it to preserve angular momentum.
PolyPIC \cite{fu17} improves the energy and vorticity conservation of APIC by considering more velocity modes.

\citet{ando12} detect and preserve thin fluid sheets, which are reconstructed using anisotropic kernels.
\citet{um14} use sub-grid particle correction for better particle distribution.
The band method \cite{ferstl16band} keep particles only within a narrow band of the liquid surface to improve performance.
\citet{sato18} extend the band method and add particle correction based on \cite{ando12} to better distribute particles near the surface.
\citet{takahashi19} simulate viscous materials based on APIC with strong two-way coupling with solids. They apply position correction based on density constraints, using a purely Lagrangian approach (SPH).

From the mass conservation law, \citet{kugelstadt19} derive a pressure Poisson equation which takes density deviation into account.
They add a density correction step that recovers fluid volume, which involves solving an additional Poisson equation.
Density correction was previously performed when using the so-called unilateral incompressibility constraint, which was used for free-flowing granular materials \cite{narain10} and animating splashing liquids \cite{gerszewski13}.

Power particles \cite{degoes15power} considers particles as having volume, and the fluid domain is partitioned as a power diagram.
The particle volumes can be prescribed, which enables controlling the fluid's volume and leads to better particle distribution.
Power PIC \cite{qu22} improves the performance of \cite{degoes15power} by reformulating the problem as a transportation problem, which is solved efficiently using Sinkhorn's iterative algorithm.
In a 2-phase simulation, an estimate of the surface is used as an air occupancy baseline for slack air variables. The variables fill the gaps between the prescribed volume of the fluid particles and the volume of the cells.

\citet{elcott07} rewrite the Euler equations in terms of vorticity instead of pressure. They use discrete exterior calculus (DEC) for discretization, which is readily applied to meshes.
DEC theory guarantees that certain properties hold, which arbitrary discretizations \cite{ando13} cannot.
\citet{ando15} apply the same approach to a regular grid, which reduces DEC to finite difference operators. It coincides with the MAC grid discretization \cite{bridson15fluid}, which is justified by DEC theory that associates k-forms with specific elements of the grid.
The equations are derived from an energy, using a variational approach \cite{batty07}.
\citet{dewitt12} represent vorticity and velocity using a basis of Laplacian eigenfunctions, which admits closed-form solutions on simple domains.
Since Laplacian eigenfunctions correspond with spatial scales of vorticity, basis coefficients can be seen as a discrete spectrum of vorticity.

\section{Background} \label{sec.background}
We model a fluid in a domain $\Omega \subset \R^d$, $d \in \{2, 3\}$, using the Navier--Stokes equations, which can be written for inviscid, incompressible flows as:
\begin{subequations}
\begin{align}
\rho \frac{\Dif u}{\Dif t} &= -\nabla p + f \\
\nabla \cdot u &= 0 \quad \label{eq.div_free} ,
\end{align}
\end{subequations}
where $u, f \in \R^d$ and $p, \rho \in \R$ denote velocity, external forces, pressure, and density.
$\frac{\Dif}{\Dif t}$ denotes the material derivative.
\Cref{eq.div_free} enforces the velocity to be divergence-free, which ensures incompressibility.

The FLIP method is a hybrid discretization method that combines the Eulerian and Lagrangian views.
The domain is discretized using a regular (square or cubic) grid, and the fluid is discretized using particles.
The FLIP algorithm alternates between the views, solving pressure over the grid, which is more accurate, and advecting quantities using particles, which is more robust.
Discretization inaccuracies accumulate over time and violate incompressibility, and a correction step is needed to maintain it.
The algorithm is listed in \cref{alg.flip} for a single time step.
\begin{algorithm}[t]
    \DontPrintSemicolon

    Transfer velocity from particles to grid \tcp*[h]{ particles are at $\bar{x}$ } \;
    Apply external forces to grid \;
    Solve for pressure \;
    Transfer velocity from grid to particles \;
    Advect particles \tcp*[h]{ new positions are at $\hat{x}$ } \;
    Correct particle positions \tcp*[h]{ new positions are at $x$ } \;

    \caption{A time step of the FLIP method} \label{alg.flip}
\end{algorithm}
The optional correction step is missing in the original FLIP method, and it is implemented differently by IDP and our method.

\section{Discrete Incompressibility} \label{sec.discrete_inc}
We propose a definition for discrete incompressibility based on grid resolution.
We define discrete density as the number of particles in a grid cell, and we denote its units by ppc, which stands for particles per cell.
We initially propose the following simple condition for discrete incompressibility, which we will relax in \cref{sec.surface_bubble}:
keep a constant number $\mu \in \Z$ of particles in a fluid cell throughout the simulation.
$\mu$ is given, and it is usually based on the initial fluid state.

Preserving the condition is done in the correction step that is described in the following, starting with notations.
Let $\mathcal{C}$ be the set of grid cells that cover $\Omega$.
We associate markings with grid cells, which describe their characteristics. Each marking has its subset of cells that are marked with it.
Initially, we use the disjoint subsets $\mathcal{C}_\text{empty}$, $\mathcal{C}_\text{solid}$, and $\mathcal{C}_\text{fluid}$ to mark cells that are empty, part of a solid, or contain fluid.
For each cell $c$, we define a set $\gamma_c$ of the indices of the particles that are in the cell.

Given a set of $n$ particles and their positions $\hat{x} \in \Omega^n$ after advection, we would like to solve for new particle positions $x \in \Omega^n$ that are close to $\hat{x}$ but preserve incompressibility.
We will refer to $\hat{x}$ as ideal positions.
We define a cost function for closeness that penalizes the distance between two points $q, r \in \Omega$:
\[ \sigma_\text{obj}(q, r) \coloneq \norm{ q - r }_2^2 \ . \]
This leads to the following problem:
\begin{mini!}
{ x } { \sum_{j=1}^{n} \sigma_\text{obj}( x_j, \hat{x}_j ) \label{eq.problem.a} } {\label{eq.problem}} {}
\addConstraint{ \abs{\gamma_c} }{ = \mu \ , \quad  \forall c \in \mathcal{C}_\text{fluid} } \label{eq.problem.b}
\addConstraint{ \abs{\gamma_c} }{ \le 0 \ , \quad  \forall c \in \mathcal{C}_\text{solid} } \label{eq.problem.c} \ ,
\end{mini!}
where $x_j, \hat{x}_j \in \Omega$ are the new and ideal positions of the $j$th particle, and $\abs{\gamma_c}$ denotes the number of particles in cell $c$.
Cell markings $\mathcal{C}_\text{empty}$, $\mathcal{C}_\text{fluid}$ and the sets $\gamma$ are determined by $x$.
That is, $\mathcal{C}_\text{solid}$ is determined by objects in the scene, and \cref{eq.problem.c} ensures that it does not contain particles. $\mathcal{C}_\text{empty}$, on the other hand, is determined by particle positions, and it does not require a constraint.

\subsection{Surface and Bubbles} \label{sec.surface_bubble}
\begin{figure}
    \centering
    \includegraphics[width=1\columnwidth]{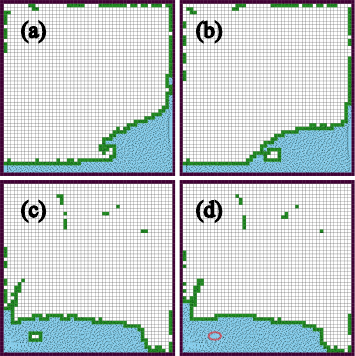}
    \caption{
    The surface (green cells) of a breaking wave (a) closes upon itself, creating an air pocket (b).
    The air pocket shrinks (c) until surface cells do not have empty neighbors anymore and become inner cells (d). These former surface cells may have less than $\mu$ particles, and in such a case we say that they contain air bubbles (an example is circled in red).
    }
    \label{fig.bubble}
\end{figure}
The discrete incompressibility condition in the previous section is too restrictive, and we relax it in a reasonable way.

We will use two types of neighborhoods for a grid cell.
The first is a von Neumann neighborhood, which refers to 4-connectivity in 2D and 8-connectivity in 3D (axis-aligned directions).
The second is a Moore neighborhood, which refers to 8-connectivity in 2D and 26-connectivity in 3D (two grid cells are neighbors if the Chebyshev distance between their centers is 1).

\begin{definition}[surface] \label{def.surface}
Given cell markings, we define a fluid cell as surface if it neighbors in a Moore neighborhood a cell that is not a fluid cell and is not on the domain boundary.
We partition $\mathcal{C}_\text{fluid}$ into surface cells $\mathcal{C}_\text{surface}$ and inner cells $\mathcal{C}_\text{inner}$.
\end{definition}

In the definition, we used only the domain boundary, distinguishing between solid cells that are static (tank walls) and solid cells that may move (\cref{sec.soilds}).

We start with motivation.
Consider a common setup where the fluid initially occupies a rectangular shape of 2D grid cells with density $\mu$ ppc. The incompressibility conditions from the previous sections means that the fluid will flow and change in full cells only, i.e., a cell from the surface with four particles will become empty, and an empty cell near the surface will gain four particles. This rigidity will cause particles to lose their resolution and behave as a unit or a single particle within a cell.
To allow shape flexibility and individual particle movement between cells, we permit surface cells to have $\le \mu$ particles.
In the next section, we will extend the relaxation on surface cells to another layer of cells incident to the surface to make the problem easier.

\begin{definition}[bubble]
An inner cell $c \in \mathcal{C}_\text{inner}$ is said to contain an air bubble if it is not empty and the number of its particles is less than $\mu$.
\end{definition}

In certain cases, such as a breaking wave, the surface curls and folds upon itself. This leads to a moment in time, where a surface cell suddenly becomes an inner cell since it no longer has any incident empty cells.
This former surface cell may be only partially filled (has less than $\mu$ particles), which would violate the incompressibility condition.

To address this, we allow inner fluid cells to keep air bubbles. Specifically, instead of demanding from inner cells to have at least $\mu$ ppc, we demand that they will not lose particles (since becoming an inner cell) and will have at least as many particles as they had in the previous iteration.
See \cref{fig.bubble} for an illustration.

We apply the relaxed conditions to \cref{eq.problem}:
\begin{mini!}
{ x } { \sum_{j=1}^{n} \sigma_\text{obj}( x_j, \hat{x}_j ) \label{eq.problem2.a} } {\label{eq.problem2}} {}
\addConstraint{ \abs{\gamma_c} }{ \le \mu \ , \quad \forall c \in \mathcal{C}_\text{surface} } \label{eq.problem2.b}
\addConstraint{ \abs{\bar{\gamma}_c} \le \abs{\gamma_c} }{ \le \mu \ , \quad \forall c \in \mathcal{C}_\text{inner} } \label{eq.problem2.c}
\addConstraint{ \abs{\gamma_c} }{ \le 0 \ , \quad  \forall c \in \mathcal{C}_\text{solid} } \label{eq.problem2.d} \ ,
\end{mini!}
where $\bar{\gamma}$ refers to $\gamma$ from the previous iteration.
Note that while the relaxed constraints allow reasonable expansion, they do not allow compression: there can be at most $\mu$ ppc in a cell.

\section{Grid Movement} \label{sec.grid_move}
The problem in \cref{eq.problem2} is hard. To make it more manageable, we reformulate it in terms of grid movement.

We start by limiting the size of the simulation time step such that every particle does not move more than one cell (i.e., to the local Moore neighborhood).
$\hat{x}$ will refer to particle positions after advection with the updated time step.

We reduce the possible grid movements of a particle to a set of grid directions in a von Neumann neighborhood. Namely, define the set $\mathcal{D}$ that consists of the columns of the $d \times d$ identity matrix and their negation, along with the zero vector that signifies staying in the same cell.
For example, in 2D:
\begin{equation} \label{eq.directions}
\mathcal{D} \coloneq \left\{ \mathcal{D}_i \right\}_{i=1}^5
= \left\{
\begin{pmatrix*}[r] 1\\ 0 \end{pmatrix*},
\begin{pmatrix*}[r] 0\\ 1 \end{pmatrix*},
\begin{pmatrix*}[r] -1\\ 0 \end{pmatrix*},
\begin{pmatrix*}[r] 0\\ -1 \end{pmatrix*},
\begin{pmatrix*}[r] 0\\ 0 \end{pmatrix*}
\right\} \ .
\end{equation}
Let $m \coloneq \abs{ \mathcal{D} }$ (m=5 in 2D, and m=7 in 3D).
Experimentally, we found that using a Moore neighborhood (9 directions in 2D and 27 directions in 3D) did not make a significant difference.

Instead of solving for particle positions, we solve for particle grid movements. For each particle, we choose one direction from $\mathcal{D}$.
Let $b \in \Z_2^{ m \times n }$ be a binary matrix that chooses a direction for each particle.
The $j$th column is assigned to the $j$th particle, and it contains a single nonzero in the entry that corresponds to the particle's chosen direction.

Let $\phi_\text{center}(c) \in \Omega$ return the center of a cell $c$.
Let $\phi_\text{close}(q, c) \in \Omega$ return the point in a cell $c$ that is closest to a point $q \in \Omega$.
That is, the $k$th component of $\phi_\text{close}(q, c)$ is
\begin{equation} \label{eq.best_cell_pos}
( \phi_\text{close}(q, c) )_k \coloneq \begin{cases}
    q_k & \text{if } \abs{v_k} < 0.5 \\
    \lceil ( \phi_\text{center}(c) )_k \rceil - \epsilon & \text{if } v_k \ge 0.5 \\
    \lfloor ( \phi_\text{center}(c) )_k \rfloor + \epsilon & \text{if } v_k \le -0.5
    \end{cases} \ ,
\end{equation}
where $v \coloneq q - \phi_\text{center}(c)$, and $\epsilon$ (=0.01) is a small margin.

Let $\bar{x} \in \Omega^n$ be the particle positions at the end of the previous iteration.
The matrix entry $b_{ij}$ corresponds to a possible grid movement $\mathcal{D}_i \in \R^d$ for the $j$th particle. We associate a specific particle position with this entry, which is the optimal position in the cell that the particle will end up in w.r.t. its ideal position and $\sigma_{obj}$:
\[ \xi_{ij} \coloneq \phi_\text{close}( \hat{x}_j, cell(\bar{x}_j + \mathcal{D}_i) ) \ , \]
where $\bar{x}_j$ denotes the position of the $j$th particle in the previous iteration, and $\cell( \cdot )$ returns the grid cell that contains a given point in the domain.
Note that all positions $\xi \in \Omega^{ m \times n }$ are known.
See \cref{fig.xi} for illustration.
\begin{figure}
    \centering
    \includegraphics[width=.5\columnwidth]{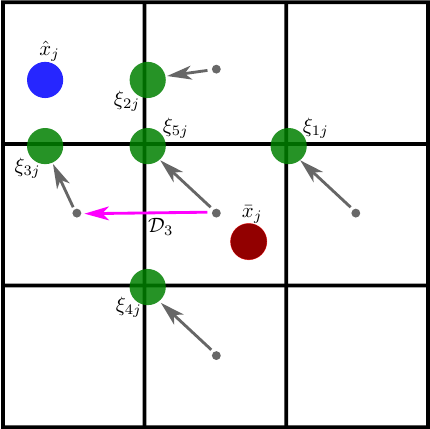}
    \caption{
    Possible positions $\xi_{ij}$ (in green) of the $j$th particle.
    A gray arrow points from the center of the cell that the particle is confined to to the closest location in that cell (up to a margin $\epsilon$) to $\hat{x}_j$ (in blue).
    If $\hat{x}_j$ was in one of the five possible cells, then the corresponding $\xi_{ij}$ to that cell would have coincided with it.
    All the possible positions $\xi_{ij}$ are known, and a solution to \cref{eq.problem3} selects one as the position $x_j$ of the particle at the end of the iteration.
    There are five cells corresponding to $\mathcal{D}$.
    One example of a discrete direction is give by the magenta arrow, which points from the center of the cell that contains $\bar{x}_j$ (in red) to the center of the neighboring cell in the discrete direction $\mathcal{D}_3$.
    }
    \label{fig.xi}
\end{figure}

For each cell $c$, we define a set $\tilde{\gamma}_c$ of index pairs for the particles that may end up in it. Each index pair is associated with a particle movement possibility, and it consists of an index of a direction and an index of a particle:
\[ \tilde{\gamma}_c \coloneq \left\{ (i, j) \mid \cell( \xi_{ij} ) = c \right\} \]

The problem becomes:
%
\begin{mini!}
{ b } { \sum_{i = 1}^m \sum_{j=1}^n b_{ij} \sigma_\text{obj}( \xi_{ij}, \hat{x}_j ) \protect\label{eq.problem3.a} }{ \protect\label{eq.problem3} }{}
\addConstraint{ 0 \le b_{ij} }{ \le 1 \ , \quad }{ \forall b_{ij} \in b  \protect\label{eq.problem3.b} }
\addConstraint{ \sum_{i = 1}^{m} b_{ij} }{ = 1 \ , \quad }{ j = 1 \dots n  \protect\label{eq.problem3.c} }
\addConstraint{ \sum_{ (i, j) \in \tilde{\gamma}_c } b_{ij} }{ \le \mu \ , \quad }{ \forall c \in \bar{\mathcal{C}}_\text{empty} \cup \bar{\mathcal{C}}_\text{surface} \protect\label{eq.problem3.d} }
\addConstraint{ \abs{\bar{\gamma}_c} \le \sum_{ (i, j) \in \tilde{\gamma}_c } b_{ij} }{ \le \mu \ , \quad }{ \forall c \in \bar{\mathcal{C}}_\text{inner}  \protect\label{eq.problem3.e} }
\addConstraint{ \sum_{ (i, j) \in \tilde{\gamma}_c } b_{ij} }{ \le 0 \ , \quad }{ \forall c \in \bar{\mathcal{C}}_\text{solid} \ . \protect\label{eq.problem3.f} }
\end{mini!}
Details:
\begin{itemize}
\item The objective in \cref{eq.problem3.a} is similar to \cref{eq.problem2.a}.
All $\sigma_\text{obj}( \xi_{ij}, \hat{x}_j )$ are known, and $b$ ensures that only selected particle movements contribute to the sum.

\item \Cref{eq.problem3.b} asserts the range of binary variables.

\item \Cref{eq.problem3.c} forces a single selected direction for each particle.

\item \Cref{eq.problem3.d} and \cref{eq.problem3.e} are similar to the incompressibility constraints \crefrange{eq.problem2.b}{eq.problem2.c}.
The sum $\sum_{ (i, j) \in \tilde{\gamma}_c } b_{ij}$ counts the particles that end up in cell $c$.
$\bar{\mathcal{C}}$ refers to the markings in the previous iteration.

\item \Cref{eq.problem3.f} is similar to \cref{eq.problem2.d}
\end{itemize}
In \cref{eq.problem3.d}, we use the surface marking from the previous iteration to relieve the need to track the surface during optimization (or formulate a constraint that handles the two cases of a surface cell remains a surface or becomes an inner cell).
This extends the relaxed condition on the surface from \cref{sec.surface_bubble} to another layer of cells incident to the surface (the condition now applies to surface cells in the previous iteration, which may belong to the layer of inner cells incident to the surface in this iteration), which is still within reason.

The problem is always feasible since $\bar{x}$ is in the solution space.
Given a solution $b^*$, the final particle position $x_j$ is set to the $\xi_{ij}$ that corresponds to its selected movement direction indicated by the $i$th entry with the single nonzero in the $j$th column of the solution $b^*$.

The problem in \cref{eq.problem3} is a linear programming problem with binary variables $b$ only (the rest of the symbols are fixed, including cell markings, index sets, and particle positions, which do not depend on $b$), which is a type of integer linear programming (ILP).
Satisfying a 0-1 ILP is one of Karp's 21 NP-complete problems.
The following proposition allows us to relax the problem to a standard linear programming (LP) with continuous variables $b \in \R^{ m \times n }$, for which there are polynomial-time solvers, and it is faster to solve in general.
\begin{restatable}{proposition}{LPrelax} \label{thm.LPrelax}
The LP relaxation of the ILP in \cref{eq.problem3}, which uses continuous variables, has the same optimal solution.
\end{restatable}
See proof in \cref{sec.proofs}.

\section{The Band Method} \label{sec.band_method}
\begin{figure}
    \centering
    \includegraphics[width=.5\columnwidth]{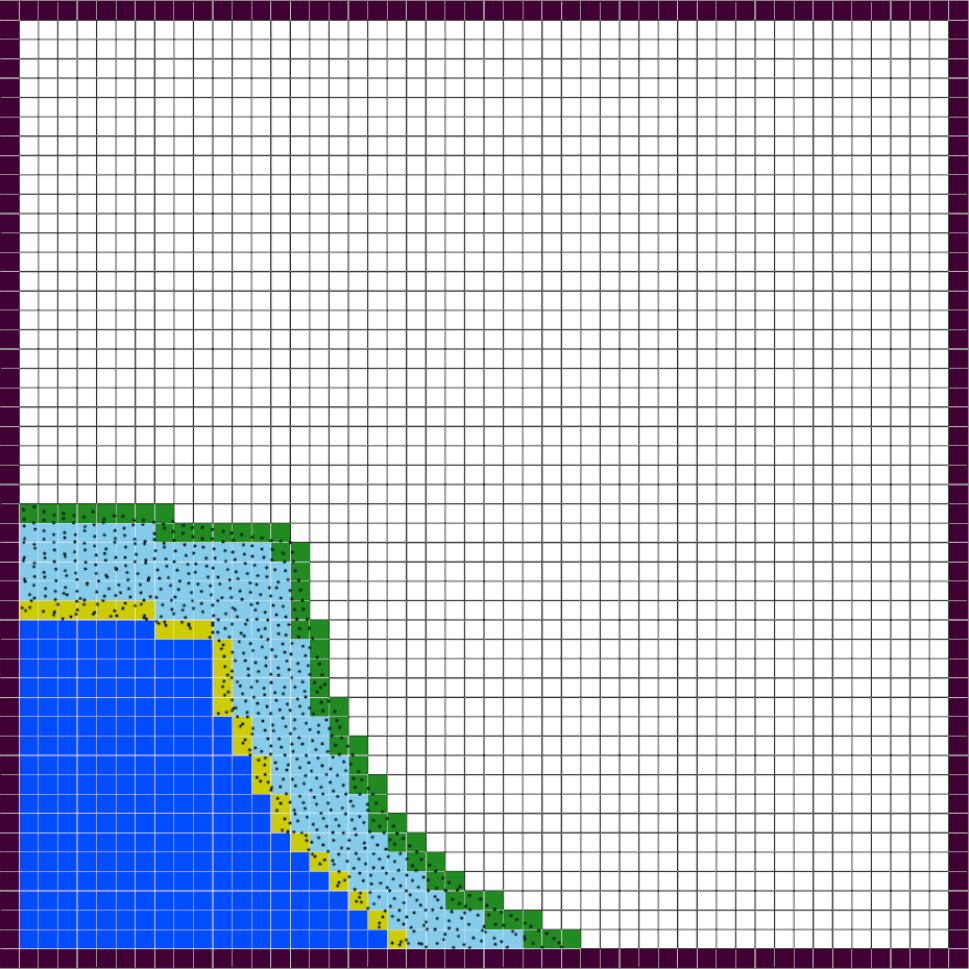}
    \caption{
    A band.
    Deep cells ($\mathcal{C}_{<-R}$) in dark blue, band interface ($\mathcal{C}_{-R}$) in yellow, surface ($\mathcal{C}_0$) in green, and the rest of the band ($\mathcal{C}_{-R < \beta < 0}$) in light blue.
    }
    \label{fig.band}
\end{figure}
Solving the LP in \cref{eq.problem3} does not scale well, and for large 3D grids, we propose an incompressible variant of the band method \cite{ferstl16band}.
The method uses only a fraction of the number of particles, which directly affects the size of the LP.

The motivation for the band method is based on the observation that most of the interesting, complex behavior of a fluid happens close to the surface.
FLIP uses particles to reduce numerical dissipation and keep the simulation lively.
Based on the observation above, particles at deeper levels of the fluid do not contribute much to visual appearance.
Leveraging that, the method maintains only a narrow band of particles near the fluid surface and uses an Eulerian-grid approach to simulate the rest of the fluid.
The grid velocity at each iteration is determined by a combination of the two.

To maintain fluid density for incompressibility, we need to supervise the number of particles that enter and leave the band.
Furthermore, while the method in \cite{ferstl16band} uses an approximate distance from the surface to define the band, our discrete approach that uses hard constraints requires a more careful estimate.

\begin{definition}[depth] \label{def.depth}
Each fluid cell is assigned a depth $\beta \in \Z$ that represents its discrete (signed) distance from the surface, and it is derived from the state of the fluid (particle positions) at a time step.
The depth is assigned recursively in a breadth-first-search manner.
Surface cells are assigned depth $\beta=0$.
Their neighboring fluid cells in a von Neumann neighborhood are assigned $\beta=-1$.
The unassigned fluid neighbors of the $\beta=-1$ cells are assigned one level lower, $\beta=-2$, and so on until all fluid cells are assigned a depth.
All non-fluid cell are assigned an arbitrary positive number (e.g., 1) as depth.
\end{definition}

Let $R + 1 \in \mathbb{N}$ be the thickness of the particle band.
We define the cells at depth $-R \le \beta \le 0$ to be within the band.
We call cells at depth $\beta = -R$ \textit{band interface} and cells at depth $\beta < -R$ \textit{deep}.
We add marking subsets to distinguish between parts and depth levels of the fluid.
Denote by $\mathcal{C}_{\beta=k}$ or simply $\mathcal{C}_k$ (when clear from the context) the set of cells at depth $\beta = k$. We extend the notation to a range of depth levels, e.g., $\mathcal{C}_{<-R}$ will denote deep cells.
See \cref{fig.band} for illustration.

\begin{algorithm}[t]
    \DontPrintSemicolon

    Transfer velocity from particles to grid \textbf{and combine with current grid velocity} \tcp*[h]{ $\bar{x}$ } \;
    Apply external forces to grid \;
    Solve for pressure \;
    Transfer velocity from grid to particles \;
    \textbf{Advect grid velocity} \;
    Advect particles \tcp*[h]{ $\hat{x}$ } \;
    Correct particle positions \tcp*[h]{ $x$ } \;
    \textbf{Update cell markings} \;
    \textbf{Remove particles that reached the deep and add excess particles to the band interface} \;

    \caption{A time step of the band method} \label{alg.band}
\end{algorithm}
\Cref{alg.band} outlines the steps of the band method. Changes from \cref{alg.flip} are emphasized.
When transferring velocity from particles to the grid, the particles' velocity is copied only for cells within the band, not including the band interface.
The velocity in the rest of the grid cells remain unchanged.
When correcting the particles position, we modify our algorithm to handle the band (next sections).
Advecting grid velocity, which is needed for the part of the fluid without particles (not in the band), is done using the common semi-Lagrangian approach \cite{stam99}.
Cell markings $\mathcal{C}$ are updated, and the step is emphasized in \cref{alg.band} to clarify that it is performed at the end of the correction step and before removing and adding particles.

Particles are limited to the band, and particles which go deep are deleted.
To maintain incompressibility, the excess of particles in the deep is moved into the band interface, as described next.
We keep track in a variable $n_{\text{deep}}$ of the number of (imaginary) particles that are in the deep, updating the variable every deletion and insertion of a particle.
The excess of particles in the deep is
\[ n_{\text{excess}} \coloneq n_{\text{deep}} - \mu \abs{ \mathcal{C}_{<-R} } \ . \]
When $n_{\text{excess}} > 0$, we add $n_{\text{excess}}$ particles to the band interface. We randomly iterate the cells in the band interface and fill them up to $\mu$ with remaining excess particles.
Each added particle is positioned randomly within a cell, and its velocity is interpolated from the grid velocity.
Note that there is always space in the band interface for excess particles from the deep since we constrain the number of movements into and out of the band interface (next sections).

In the next sections, we offer three variants to control the movements into and out of the band interface, where each is faster than the former.
Foundation and concepts are laid out through the sections, culminating in the fastest variant.

\subsection{A Band Constraint} \label{sec.band_con}
We maintain incompressibility by controlling the comings and goings of particles through the band interface. We want the number of particles that move from a shallower depth level ($\beta = 1 - R$) to the band interface ($\beta = -R$) to be equal to the number of particles that move in the other direction.

We define two sets of index pairs of particle movement possibilities, into and out of the band interface (from and into a shallower level):
\begin{align*}
\tilde{\gamma}_\text{in} & \coloneq \left\{ (i,j) \mid \cell( \bar{x}_j ) \in \mathcal{\bar{C}}_{1-R} \ , \  \cell( \xi_{ij} ) \in \mathcal{\bar{C}}_{-R} \right\} \\
\tilde{\gamma}_\text{out} & \coloneq \left\{ (i,j) \mid \cell( \bar{x}_j ) \in \mathcal{\bar{C}}_{-R} \ , \  \cell( \xi_{ij} ) \in \mathcal{\bar{C}}_{1-R} \right\} \ .
\end{align*}

The deep may contain air bubbles from cells that carried bubbles while moving to the deep.
We would like to allow bubbles in the band interface and deep to fill up.
Let $\alpha_{-R}$, $\alpha_{<-R}$  be the total amounts of air bubbles (number of missing particles) in the band interface and deep, which can be calculated from the number of cells and particles:
\begin{align*}
\alpha_{-R} & \coloneq \mu \abs{ \mathcal{\bar{C}}_{-R} } - \sum_{ c \in \mathcal{\bar{C}}_{-R} } \abs{ \bar{\gamma}_c } \\
\alpha_{<-R} & \coloneq \mu \abs{ \mathcal{\bar{C}}_{<-R} } - n_{\text{deep}} \ .
\end{align*}
We express the conditions above as an additional constraint to \cref{eq.problem3}:
\begin{equation} \label{eq.band_con}
0 \le \sum_{ (i, j) \in \tilde{\gamma}_\text{in} } b_{ij} - \sum_{ (i, j) \in \tilde{\gamma}_\text{out} } b_{ij} \le \alpha_{\le -R} \ .
\end{equation}
Where $\alpha_{\le -R} \coloneq \alpha_{-R} + \alpha_{<-R}$.
We also update \cref{eq.problem3.d} and \cref{eq.problem3.e} to use the band markings:
\begin{subequations} \label{eq.band_con_update}
\begin{align}
\sum_{ (i, j) \in \tilde{\gamma}_c } b_{ij} & \le \mu \ , \quad \forall c \in \bar{\mathcal{C}}_\text{empty} \cup \bar{\mathcal{C}}_\text{surface} \cup \bar{\mathcal{C}}_{-R} \label{eq.band_con_update.a} \\
\abs{\bar{\gamma}_c} \le \sum_{ (i, j) \in \tilde{\gamma}_c } b_{ij} & \le \mu \ , \quad \forall c \in \bar{\mathcal{C}}_{-R < \beta < 0} \ ,
\end{align}
\end{subequations}
where we allow the band interface the same flexibility as the surface (to lose particles) since excess deep particles will be added back to it.
We put no constraint on deep cells due to the particle deletion step.

The correction step is done as before by solving the updated problem in \cref{eq.problem3} for $b$ and updating $x$ accordingly.
Unlike the local, sparse constraints in \cref{eq.problem3}, the band constraint is global and dense since it encompasses and ties together particle movements along the band interface.
Moreover, the system matrix may not be totally unimodular anymore, and the ILP problem cannot be relaxed.
In some scenes, these increase the solver time such that it is not much better than not using a band ($R = \infty$).
In the next sections, we offer faster alternatives.

\subsection{A One-Way Band Constraint} \label{sec.band_con2}
One way to shorten the constraint in \cref{eq.band_con} is to determine first the number of particles that go into and out of the band, and based on that, constrain only the number of particles in the direction with the greater flow.
We do this in two steps, solving an LP in the first step and an ILP in the second.

First, we solve \cref{eq.problem3} as is (without an additional band constraint), getting an optimal solution $b^*$, and we do not update $x$ yet. From these particle movements, denote the number of particles that go into and out of the band interface by
\begin{align*}
n^*_\text{in} & \coloneq \sum_{ (i, j) \in \tilde{\gamma}_\text{in} } b^*_{ij} \\
n^*_\text{out} & \coloneq \sum_{ (i, j) \in \tilde{\gamma}_\text{out} } b^*_{ij} \ .
\end{align*}
Consider the differences
\begin{align*}
s_\text{in} & \coloneq n^*_\text{out} - n^*_\text{in} + \alpha_{\le -R} \\
s_\text{out} & \coloneq n^*_\text{in} - n^*_\text{out} \ .
\end{align*}
$s_\text{in}$ measures how much space is left in the band interface and deep, and $s_\text{out}$ measures the space in the rest of the band.
If both $s_\text{in} \ge 0$ and $s_\text{out} \ge 0$, then the movements are fine, we can update $x$ according to $b^*$ and proceed with the rest of the algorithm.
Else, there is negative space (incompressibility is violated), and we solve \cref{eq.problem3} a second time with an additional constraint, depending on which space is negative.

If $s_\text{in} < 0$, then too many particles moved into the band interface, and we need to limit them.
We fix the movements of all $n^*_\text{out}$ particles that moved from the band interface to the rest of the band and block the rest of the movements in $\tilde{\gamma}_\text{out}$.
In addition, we prevent movement into the band from particles in $\bar{\mathcal{C}}_{1 - R}$ that do not move into the band interface in $b^*$.
We end up with a constraint that selects $n^*_\text{out} + \alpha_{\le -R}$ particles from the particles that moved into the band in $b^*$:
\begin{subequations} \label{eq.band_con2}
\begin{align}
b_{ij} &= b^*_{ij} \ , \quad && \forall (i,j) \in \tilde{\gamma}_\text{out}  \label{eq.band_con2.a} \\
b_{ij} &= 0 \ , \quad && \forall (i,j) \in \tilde{\gamma}_\text{in} \ , \ b^*_{ij} = 0 \label{eq.band_con2.b} \\
\sum_{ (i, j) \in \tilde{\gamma}_\text{in} } b_{ij} &= n^*_\text{out} + \alpha_{\le -R} \ . \label{eq.band_con2.c}
\end{align}
\end{subequations}
When setting a movement $b_{ij}$ of the $j$th particle in \cref{eq.band_con2.a} to one, due to \cref{eq.problem3.c}, we can also set the rest of the particle's movements to zero: $\forall k \ne i, b_{kj} = 0$.
Note that \cref{eq.band_con2.a} and \cref{eq.band_con2.b} merely eliminate variables from the system, which leaves a single constraint \cref{eq.band_con2.c} that sets the number of particles that enter the band interface.
Due to $b^*$, which moves more particles than required, we know that the problem in \cref{eq.problem3} with the additional constraint \cref{eq.band_con2} is feasible.

Else, $s_\text{out} < 0$, and we need to limit the number of particles that move out of the band. Similar to \cref{eq.band_con2}, this is expressed as
\begin{subequations} \label{eq.band_con3}
\begin{align}
b_{ij} &= b^*_{ij} \ , \quad && \forall (i,j) \in \tilde{\gamma}_\text{in} \label{eq.band_con3.a} \\
b_{ij} &= 0 \ , \quad && \forall (i,j) \in \tilde{\gamma}_\text{out} \ , \ b^*_{ij} = 0 \label{eq.band_con3.b} \\
\sum_{ (i, j) \in \tilde{\gamma}_\text{out} } b_{ij} &= n^*_\text{in} \ , \label{eq.band_con3.c}
\end{align}
\end{subequations}
where \cref{eq.band_con3.a} fixes the variables of movements into the band interface, \cref{eq.band_con3.b} prevents movements out of the band interface that do not occur in $b^*$, and \cref{eq.band_con3.c} sets the number of particles that leave the band interface

To summarize, in the first step we solve \cref{eq.problem3}.
If needed, we perform a second step, where we solve \cref{eq.problem3} again using the one-way band constraints in \cref{eq.band_con2} or \cref{eq.band_con3}.
After the steps, we update $x$ and proceed with the rest of the algorithm.
See \cref{fig.path} for an example.

The one-way band constraint is still dense, and the ILP still cannot be relaxed. However, since we reduce variables and simplify the band problem, it becomes significantly faster to solve than adding the full band constraint in \cref{eq.band_con}.

\subsection{Flow Along Paths} \label{sec.push_part}
\begin{figure*}
    \centering
    \includegraphics[width=.8\textwidth]{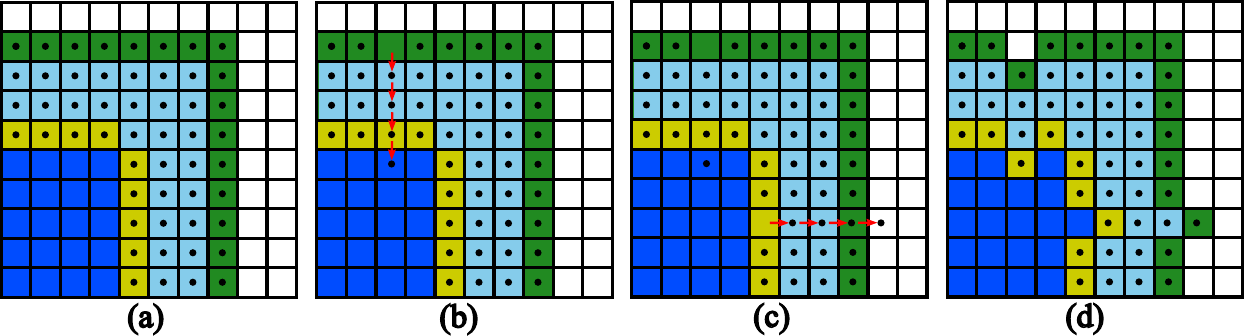}
    \caption{
    An example of correcting the band (1 ppc).
    (a) The beginning of the iteration; see \cref{fig.band} for the color code of the cells.
    (b) First step, solving the LP in \cref{eq.problem3}.
    Four particles move into the cell below them; red arrows indicate former cells.
    One of the particles moved into the band interface:
    $n^*_\text{in} = 1$, $n^*_\text{out} = 0$, $s_\text{in} = -1$, $s_\text{out} = 1$.
    Since $s_\text{in} < 0$, incompressibility is violated, and we need to perform a second step to correct it.
    If we use the variant in \cref{sec.band_con2}, then we solve \cref{eq.problem3} again---based on the particle positions in (a)---fixing the movements in \crefrange{eq.band_con2.a}{eq.band_con2.b} and adding the constraint in $\cref{eq.band_con2.c}$ to set the number of particles that move into the band (zero). There will be no movement of particles between cells (only within the cells) compared to (a); such movement will occur only when particles also leave the band interface in the first step.
    (c) If we use the variant of the band method in \cref{sec.push_part}, then a path from the band interface to the surface is found.
    Reverting the vertical path in (b) is always an option. Instead, the indicated horizontal path is selected, and particles are pushed along it.
    (d) Cell markings are updated. In this case, there is no need to remove particles that reached the deep or to fill the band interface with excess particles from the deep (\cref{alg.band}).
    }
    \label{fig.path}
\end{figure*}
The second step of the one-way band constraint approach can be viewed as correcting the incompressibility in the band interface and the deep after the first step.
We suggest a cheaper way to perform the correction, which does not require solving an ILP.

We perform the same first step as in \cref{sec.band_con2} and solve \cref{eq.problem3} as is (an LP without additional band constraints), this time updating the particles positions $x$ according to $b^*$.
If $s_\text{in}$ or $s_\text{out}$ is negative, then too many particles flowed into or out of the band interface. To correct that, we move some of them along grid paths in the required direction.

If $s_\text{in}$ is negative, then we need to move $n_\text{move} \coloneq -s_\text{in}$ particles out of the band interface.
Else, if $s_\text{out}$ is negative, then we need to move $n_\text{move} \coloneq -s_\text{out}$ particles into the band interface.
Otherwise, correction is not necessary.

We limit the $j$th particle's movement to a single cell (in a von Neumann neighborhood) relative to its position in the last iteration ($\bar{x}_j$).
To maintain the incompressibility constraint, a particle can move into cell $c$ only if it has space ($\abs{\gamma_c} < \mu$, where $\gamma$ reflects the state of the updated $x$). If it does not, then another particle needs to move out from $c$ beforehand.
This means that a chain of particles needs to be moved along a grid path, starting from a cell that has the flexibility to lose a particle---a surface, a band interface, or a former empty cell (see \cref{eq.band_con_update.a}).
We need to find $n_\text{move}$ such paths.

We formulate this as a minimum-cost flow problem (MCFP) in a graph.
The grid cells are designated as graph vertices, and possible particle movements are designated as graph edges with capacity one.
We will use multiple sources and sinks, denoted $\mathcal{C}_\text{source}$ and $\mathcal{C}_\text{sink}$.
The location of sources and sinks depend on the flow direction---into or out of the band interface.
In the in direction, surface cells are sources, and band interface and deep cells are sinks.
In the out direction, interface cells are sources, and the rest of the band (surface and inner cells with bubbles) and empty cells are sinks.

The cost of an edge that represents a possible movement of the $j$th particle  into a cell $c$ is the cost of its (optimal) position in $c$ minus the cost of its current position w.r.t. its ideal position:
\begin{equation} \label{eq.edge}
\sigma_\text{edge}(j, c) \coloneq \sigma_\text{obj}( \phi_\text{close}( \hat{x}_j, c ), \hat{x}_j ) - \sigma_\text{obj}( x_j, \hat{x}_j ) \ .
\end{equation}
We limit the movement of the $j$th particle to its cell in the previous iteration and the cells neighboring that cell (in a von Neumann neighborhood).

\begin{algorithm}[t]
    \DontPrintSemicolon
    \SetArgSty{}
    \KwOut{A list $P$ of Paths and an array $J$ of $\abs{\mathcal{C}}$ edges to parents}
    Let $b_{fin}$ be an array of $\abs{\mathcal{C}}$ flags, $\sigma$ be an array of $\abs{\mathcal{C}}$ costs, and $Q$ be a priority queue \;
    \lFor{ $c \in \mathcal{C}$ }{ $J[c], b_{fin}[c], \sigma[c] \leftarrow$ none, false, $\infty$ \nllabel{alg.find.init_arrays} }
    \For( \tcp*[h]{initialize $Q$} \nllabel{alg.find.init_q} ){ $c \in \mathcal{C}_\text{source}$ }{
        \lIf{ $\abs{\gamma_c} = 0$ }{ continue \tcp*[h]{no particles} }
        $a \leftarrow$ Node( cost=0, cell=$c$, edge=ROOT, root=$c$ ) \;
        $Q$.enqueue( $a$ ) \tcp*[h]{using $a$.cost as key} \nllabel{alg.find.enqueue} \;
        $\sigma[c] \leftarrow 0$ \;
    }
    \While( \tcp*[h]{main loop} ){ not $Q$.empty() } {
        $a \leftarrow$ $Q$.dequeue() \tcp*[h]{lowest cost} \;
        $c \leftarrow a$.cell \;
        \lIf( \tcp*[h]{visited} ){ $J[c] \notin \left\{ \text{none, ROOT} \right\}$ \nllabel{alg.find.visited} }{ continue }
        \lIf( \tcp*[h]{finished} ){ $b_{fin}[a.\text{root}]$ }{ continue }
        $J[c] \leftarrow a$.edge \tcp*[h]{assign an edge to parent} \nllabel{alg.find.set_parent} \;
        \If( \tcp*[h]{a sink with space} ){ $c \in \mathcal{C}_\text{sink}$ and $\abs{ \gamma_c } < \mu$ \nllabel{alg.find.sink} }{
            $b_{fin}[a.\text{root}]$ = true \;
            $P$.add( Path( edge=$a$.edge, sink=$c$ ) ) \;
            continue \;
        }
        \For( \tcp*[h]{without the 0 vector} ){ $d \in \mathcal{D} \setminus \{ 0 \}$ \nllabel{alg.find.neighbors} }{
            $c' \leftarrow \cell( \phi_\text{center}(c) + d )$ \;
            \lIf( \tcp*[h]{visited or a source} ){ $J[c'] \ne$ none }{ continue }
            $j \leftarrow$ best\_edge( $c$, $c'$ ) \;
            \lIf( \tcp*[h]{no edge} ){ $j =$ none }{ continue }
            $t \leftarrow$ $\sigma[c]$ + $\sigma_\text{edge}(j, c')$ \tcp*[h]{total cost from source} \;
            \lIf( \tcp*[h]{is not better} ){ $t \ge \sigma[c']$ }{ continue }
            $\sigma[c'] \leftarrow t$ \;
            $Q$.enqueue( cost=t, cell=$c'$, edge=j, root=$a.root$ ) ) \;
        }
    }
    \caption{Find paths} \label{alg.find}
\end{algorithm}
%
\begin{algorithm}[t]
    \DontPrintSemicolon
    \SetArgSty{}
    \KwIn{Two cells $c, c'$}
    \KwOut{An index $j$ of a particle that can move from $c$ to $c'$ with the lowest cost}
    $\sigma, j \leftarrow \infty$, none \;
    \For{ $j' \in \gamma_c$ }{
        \lIf( \tcp*[h]{farther than one cell from $\bar{x}_{j'}$} ){ $c \ne \cell( \bar{x}_{j'} )$ and $c' \ne \cell( \bar{x}_{j'} )$ \nllabel{alg.edge.too_far} }{ continue }
        \lIf{ $j =$ none or $\sigma_\text{edge}(j', c') < \sigma$ }{ $\sigma, j \leftarrow \sigma_\text{edge}(j', c'), j'$ }
    }
    \caption{best\_edge( $c$, $c'$ )} \label{alg.edge}
\end{algorithm}
To solve the MCFP, we use a variant of Dijkstra's algorithm to find $n_\text{move}$ (augmenting) paths in the (residual) graph (using terms from the Ford--Fulkerson algorithm).
The algorithm is listed in \cref{alg.find}.

A path starts in a source cell and ends in a sink. The algorithm finds paths that do not share cells (or particles).
We hold in three arrays (of size $\abs{\mathcal{C}}$), $J$, $b_{fin}$, and $\sigma$, data related to the cells.
$J[c]$ is the index of a particle that represents an edge to the parent of $c$ on a path, where $c$ can belong to a single path at most.
$b_{fin}[c]$ is a (boolean) flag that indicates if the path that started at the source cell $c$ is complete.
$\sigma[c]$ is the total cost of the path that $c$ belongs to from its source to $c$.
The three arrays are initialized with none, false, and $\infty$ (using multiple assignment in line \ref{alg.find.init_arrays}).

The search for paths starts at source cells (skipping empty ones), which are pushed into a priority queue $Q$ of objects of type \textit{Node} (line \ref{alg.find.enqueue}). A \textit{Node} represents the last vertex in a path, which will probe for the next cell on the path.
The fields of \textit{Node} are: \textit{cost}---the sum of edge costs along the node's path; \textit{cell}---the cell that the vertex represents; \textit{edge}---the index of a particle that represents an edge to the cell's parent; and \textit{root}---the path's root (a source cell).
A new \textit{Node} is created using a constructor function with named arguments, and it is added to $Q$, where the field \textit{cost} is used as a key to compare elements.
The special value ROOT is used to indicate a root node (no parent).

In the main loop, the \textit{Node} with the lowest cost is dequeued.
Line \ref{alg.find.visited} checks if the cell has already been visited, i.e., if it has been dequeued already and has been assigned a parent.
$Q$ can hold \textit{Nodes} of the same cell but different parents (different possible paths). Only the \textit{Node} with the lowest cost is processed, and the rest are ignored. This is an efficient alternative for the priority queue's decrease-key method for sparse graphs.
If a path that started at the node's root has already finished, then the node is ignored.
If it is the first time that the node is visited, then it is assigned an edge to the parent on that path or a ROOT tag (line \ref{alg.find.set_parent}).
If the node is a sink with space (line \ref{alg.find.sink}), then the path is complete, and it is added to a (returned) list of completed paths $P$.

The neighboring cells of a node that were not visited yet are explored (line \ref{alg.find.neighbors}).
The best edge to a neighboring cell $c'$ is determined, and if the path to it has a lower cost, then $c'$ is enqueued.
\Cref{alg.edge} lists the algorithm that finds the best edge from cell $c$ to $c'$.
The condition in line \ref{alg.edge.too_far} checks that $c'$ is not farther than one cell from the particle's cell in the previous iteration.
It means that either the particle is currently in the same cell as in the last iteration or it is going to move to that cell.

\begin{algorithm}[t]
    \DontPrintSemicolon
    \SetArgSty{}
    \KwIn{An array $J$ of $\abs{\mathcal{C}}$ edges to parents and a Path $r$}
    $j, c \leftarrow r.$edge, $r$.sink \;
    \While{ $j \ne$ ROOT }{
        $c' \leftarrow \cell(x_j)$ \tcp*[h]{parent cell} \;
        $x_j \leftarrow \phi_\text{close}( \hat{x}_j, c )$ \tcp*[h]{move particle} \;
        $j, c \leftarrow J[c'], c'$ \tcp*[h]{predecessor} \;
    }
    \caption{Update a path} \label{alg.update}
\end{algorithm}
The paths are updated one at the time using \cref{alg.update} until $n_\text{move}$ paths are successfully updated.

Since \cref{alg.find} finds only nonintersecting paths, it may need be called more than once (with the paths updated).
Paths can be found (in the residual graph after an update) as long as the maximum flow is not reached.
The maximum flow is at least $n_\text{move}$ since it is possible to revert the particle positions induced by $b^*$ back to $\bar{x}$.
However, since the edge costs may be negative and we use Dijkstra, the resulting flow from the algorithm may not have the lowest cost.
We decided not to use a more expensive algorithm that finds the optimal cost since $n_\text{move}$ is only a small percentage of $n$, and in our experiments the results of using Dijkstra did not vary much from the alternative methods suggested in the previous sections.

See \cref{fig.path} for an example.

\section{Coupling with Solids} \label{sec.soilds}
We address incorporating solids into our framework.
We illustrate the idea on a simple scene of an object (also referred to as obstacle) free falling into water.
Before the object hits the water, its motion is only affected by gravity.
After hitting the water, the drag and buoyancy forces come into play, which decelerate the object until it reaches terminal velocity.
The affect on the fluid is expressed in the boundary conditions of the pressure equation \cite[chapter 5]{bridson15fluid}.
The pressure in grid cells that are marked as solid is set to $p = 0$, and along the solid boundary we have
\begin{equation} \label{eq.solid_boundary}
u \cdot n = u_\text{solid} \cdot n \ ,
\end{equation}
where $n$ is the boundary normal, and $u_\text{solid}$ is the solid velocity.

We discretize an object using particles, but we will keep using the term particles only to refer to fluid particles.
We allow a (non-empty) grid cell to be occupied by either fluid or a solid, which determines its marking.
As we did with fluid particles, the time step size is limited to restrict the solid from moving past the neighboring cells (Moore neighborhood).
If the object is going to occupy new cells, then the correction step decides whether the object moves or stays in place.
We consider objects that are denser than the fluid, and the movement of the object is prioritized unless incompressibility or fluid speed are violated.

Let $\mathcal{C}_\text{new\_solid}$ be the set of cells that are not marked as solid and the object intends to move into.
We modify the objective in \cref{eq.problem3.a} to use a new objective function:
\begin{mini}
{ b } { \sum_{i = 1}^m \sum_{j=1}^n b_{ij} \sigma_\text{solid\_obj}( \xi_{ij}, \hat{x}_j ) \ , }{ \label{eq.solid_problem} }{}
\end{mini}
where
\begin{equation} \label{eq.solid_penalty}
\sigma_\text{solid\_obj}(q, r) \coloneq \begin{cases}
    \lambda_\text{penalty} & \text{if } \cell( q ) \in \mathcal{C}_\text{new\_solid} \\
    \sigma_\text{obj}(q, r) & \text{else}
    \end{cases} \ .
\end{equation}
$\lambda_\text{penalty}$ (=1000) is set to a large weight to penalize particle movements into (potentially) new solid cells.

Given a solution to the modified problem, the correction step lets the object move only if none of the new particle positions $x$ are in $\mathcal{C}_\text{new\_solid}$; else the object stays in place.

For the band method in \cref{sec.push_part}, we modify $\sigma_\text{edge}$ in \cref{eq.edge} to use $\sigma_\text{solid\_obj}$ instead of $\sigma_\text{obj}$.

The definition of the fluid surface (\cref{def.surface}) considers fluid cells that touch a moving obstacle as surface. This enables flexibility in the movement of an obstacle.

\subsection{Clearing the Bottom}
\begin{figure}
    \centering
    \includegraphics[width=.8\columnwidth]{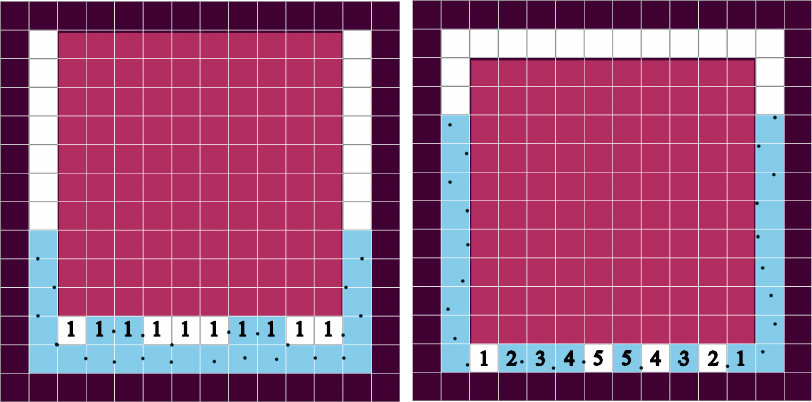}
    \caption{
    Clearing distance.
    Two frames of an obstacle (red square) falling into water (1 ppc).
    The clearing distance of cells in $\mathcal{C}_\text{new\_solid}$ is marked.
    (Left) all the cells in $\mathcal{C}_\text{new\_solid}$ have a non-obstacle neighbor below them, and their clearing distance is 1.
    (Right) Only the particle in the rightmost cell in $\mathcal{C}_\text{new\_solid}$ is guaranteed to clear the way using \cref{eq.solid_penalty}. Other particles in $\mathcal{C}_\text{new\_solid}$ require \cref{eq.solid_penalty2} to guarantee progress towards clearing the way.
    }
    \label{fig.clear_bottom}
\end{figure}
Consider particles that occupy potentially new obstacle cells and currently block the obstacle movement.
As long as there is a path in the fluid to a cell ($\notin \mathcal{C}_\text{new\_solid}$) with free space, an optimal solution will push them along the path to move them out of the obstacle's way and avoid the penalty in \cref{eq.solid_penalty}.

There is always such a path in the fluid (as long as there is space) except when the obstacle reaches the last layer of fluid before touching the bottom of the tank.
Consider such a row of cells in $\mathcal{C}_\text{new\_solid}$ between the obstacle and the bottom, where all the cells are empty except the middle one, which contains a particle. The particle has three possible cell movements: stay in the current cell, go left, or go right. Its contribution to the objective in \cref{eq.solid_problem} would be the same in each case, and nothing would motivate it to clear the way.
To address that, similar to \cref{def.depth}, we define

\begin{definition}[clearing distance] \label{def.clear_dist}
Each cell $\mathcal{C}_\text{new\_solid}$ is assigned a clearing distance that represents its discrete distance from a cell that is not in $\mathcal{C}_\text{new\_solid}$, where solid cells are ignored.
The clearing distance is assigned recursively in a breadth-first-search manner.
Cells in $\mathcal{C}_\text{new\_solid}$ with neighboring cells (in a von Neumann neighborhood) that are not in $\mathcal{C}_\text{new\_solid}$ are assigned 1.
Their unassigned neighbors are assigned one level higher, 2, and so on until all the cells in $\mathcal{C}_\text{new\_solid}$ are assigned a clearing distance.
\end{definition}

We modify \cref{eq.solid_penalty}:
\begin{equation} \label{eq.solid_penalty2}
\sigma_\text{solid\_obj}(q, r) \coloneq \begin{cases}
    \lambda_\text{penalty} \cdot \cdist( q ) & \text{if } \cell( q ) \in \mathcal{C}_\text{new\_solid} \\
    \sigma_\text{obj}(q, r) & \text{else}
    \end{cases} \ ,
\end{equation}
where $\cdist( \cdot )$ returns the clearing distance at $\cell( q )$.
This penalizes particles according to their clearing distance; see \cref{fig.clear_bottom} for illustration.

\section{Evaluation}
We implemented our method as a plugin in MantaFlow \cite{mantaflow}, using conjugate gradients to solve a Poisson equation.
We used \cite{gurobi}, selecting the dual simplex algorithm without presolve, to solve LP and ILP problems.
\\

\descrip{Measuring running time.}
The experiments were conducted on a laptop.
The running time of FLIP and IDP is dominated by solving a Poisson equation. FLIP solves one for pressure, and IDP solves an additional one for density.
The running time of our method is dominated by the solution of the LP problem.
\\

\descrip{Volume measure.}
We define the discrete volume measure of a cell $c$ based on its depth (\cref{def.depth}) as
\begin{equation} \label{eq.volume}
V_c \coloneqq \begin{cases}
    0 & c \in \mathcal{C}_{\beta > 0} \\
    \min( 1, \frac{ \norm{ \gamma_c } }{ \mu } ) & c \in \mathcal{C}_{-1 \le \beta \le 0} \\
    1 & \text{else}
    \end{cases} \ .
\end{equation}
Cells near the surface are given reasonable flexibility and are allowed to have less than $\mu$ particles.
Other fluid cells are penalized if they have less than $\mu$ particles.
All cells are penalized if they have more than $\mu$ particles.
Solid cells that contain particles are still considered pure solid, and their fluid volume is zero.

The measure used in \cite{kugelstadt19} is $\min( 1, \frac{ \norm{ \gamma_c } }{ \mu } )$ for any cell $c$. That measure is more favorable towards our method since it penalizes overflow only and overlooks air bubbles (volume inflation). According to that measure, our method preserves discrete volume perfectly.

When reporting results, we measure the volume of the whole fluid in a time step as $\frac{V}{V^*}$, where $V \coloneqq \sum_{c \in \mathcal{C}} V_c$ is the total fluid volume in a time step, and $V^*$ is how much volume should the fluid occupy. If there is no emitter in the scene, then $V^*$ is simply the initial fluid volume.
We report the range of volume percentages ($100 \frac{V}{V^*}$) over all the simulation iterations.
\\

We evaluated our method in several scenes that are described in \cref{sec.scenes}; see the accompanying video for their animation.
Statistics on volume preservation and running time are summarized in \cref{tab.vol} and \cref{tab.time}.
The grid sizes that were used in the figures and video are the ones in \cref{tab.vol}.
We compared our method with IDP \cite{kugelstadt19}, FLIP, the narrow band FLIP \cite{ferstl16band}, and Power PIC \cite{qu22}.

Power PIC has several parameters that can be crucial for its behavior, the accuracy of its particle distribution, and volume preservation.
We scaled the resolution of the transportation grid by 2 in each dimension (i.e., $\times 4$ finer than the simulation grid in 2D).
We set $\epsilon = 0.1$, $\eta = 1$, $\tau = \frac{1}{2 \sqrt[d]{\mu}}$ (e.g., $\frac{1}{4}$ for 4 ppc in 2D), and $\delta = 0.1$.
We did not cut off small coefficients from the Gaussian kernel $K$ since it increased the number of iterations due to lower accuracy.
Besides increasing the running time, large scaling of the transportation grid resulted in cracks and holes that kept forming and mending in the fluid.
Larger values for $\epsilon$ and $\delta$ disrupted the uniform particle distribution. On the other hand, effects of using smaller values ranged from the fluid becoming sluggish, exhibiting extremely high energy dissipation, up to standing still.
To summarize, Power PIC can correct the fluid's volume and particle distribution but at the risk of introducing dissipation if the changes are aggressive.

\subsection{Scenes} \label{sec.scenes}
The default settings that were used in the scenes (unless specified otherwise):
\begin{itemize}
\item The initial density is 4 ppc in 2D and 8 ppc in 3D.

\item The band method is used for our method in 3D (the variant in \cref{sec.push_part}) and only for it. The band thickness is $R=3$.

\item The maximum fluid speed is bounded.
\end{itemize}

\begin{figure}
    \centering
    \includegraphics[width=1\columnwidth]{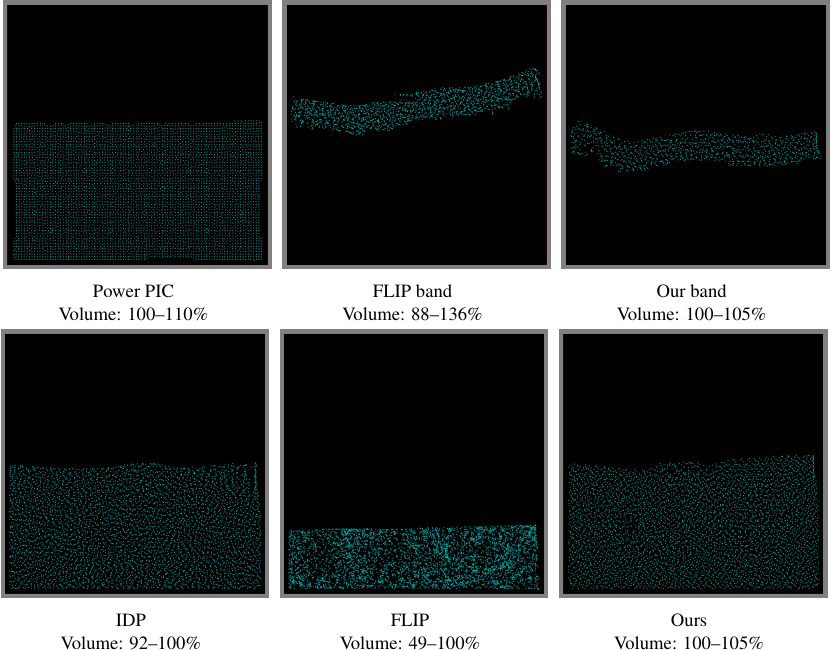}
    \caption{
    The last frame in a dam scene.
    }
    \label{fig.dam_emit_2d}
\end{figure}
\descrip{A breaking dam.}
In this scene, we perform an initial comparison of the methods behavior and volume preservation.
Some time after the dam breaks, an emitter spews a stream of water into the tank. After the emitter finishes, the total number of particles should fill exactly a half of the tank (domain); see \cref{fig.dam_emit_2d}.

\begin{figure}
    \centering
    \includegraphics[width=1\columnwidth]{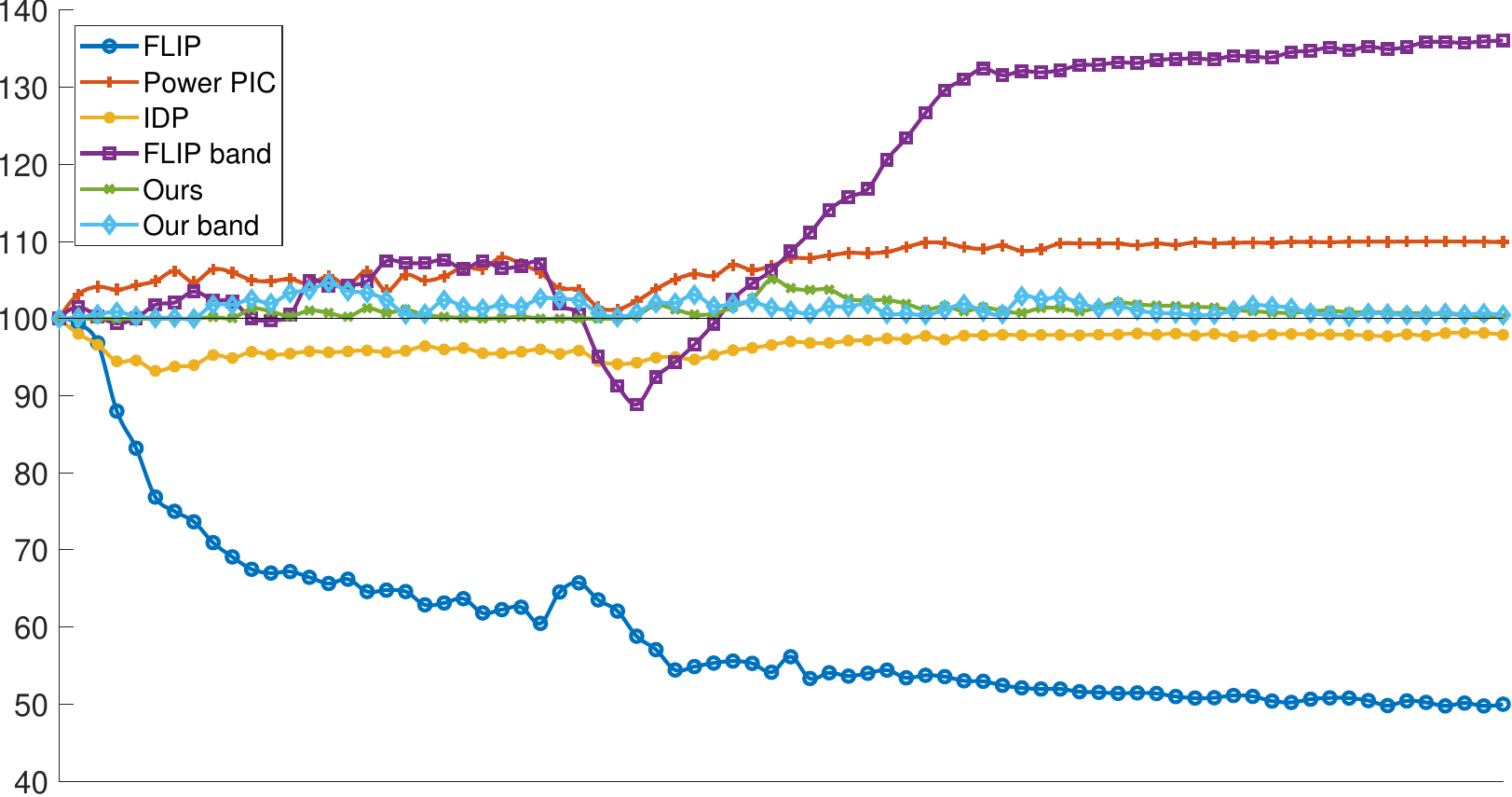}
    \caption{
    Volume (in percent) over time of the 2D dam scene.
    }
    \label{fig.plot_volume}
\end{figure}
The volume is plotted in \cref{fig.plot_volume}.
FLIP loses volume, and it is its general tendency.
The volume of the narrow band FLIP fluctuates.
Power PIC gives a nice distribution of particles and tends to preserve volume but not perfectly.
IDP tends to preserve volume but suffers some compression.
Our method uses constrained optimization and cannot lose volume. The volume may increase, however, due to air bubbles.
Our band method behaves similarly.
For both band methods we used thickness $R=6$ due to the more lively behavior of the particles compared to other scenes.

\begin{figure}
    \centering
    \includegraphics[width=1\columnwidth]{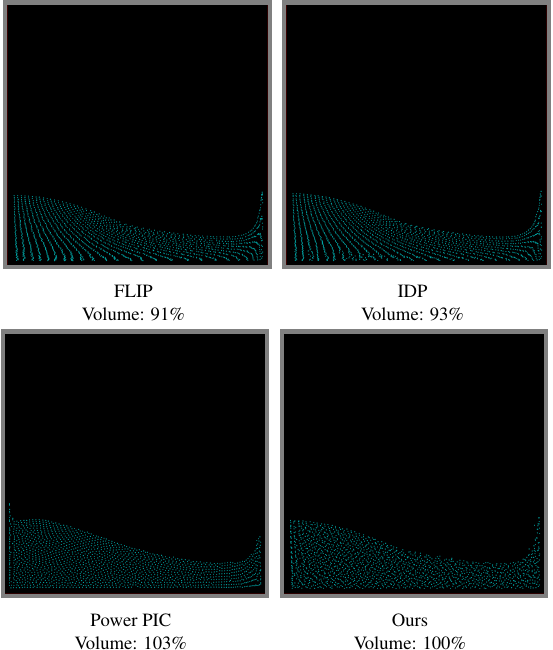}
    \caption{
    A frame from dam scene, when the water hits the right wall.
    Volume labels indicate the volume in this frame only.
    }
    \label{fig.dam_noise}
\end{figure}
\Cref{fig.dam_noise} shows a frame, where IDP keeps the clumped lines of particles and suffers volume loss.
Power PIC and our method distribute the particles, which adds noise to the fluid that reaches the surface.

\begin{figure}
    \centering
    \includegraphics[width=.8\columnwidth]{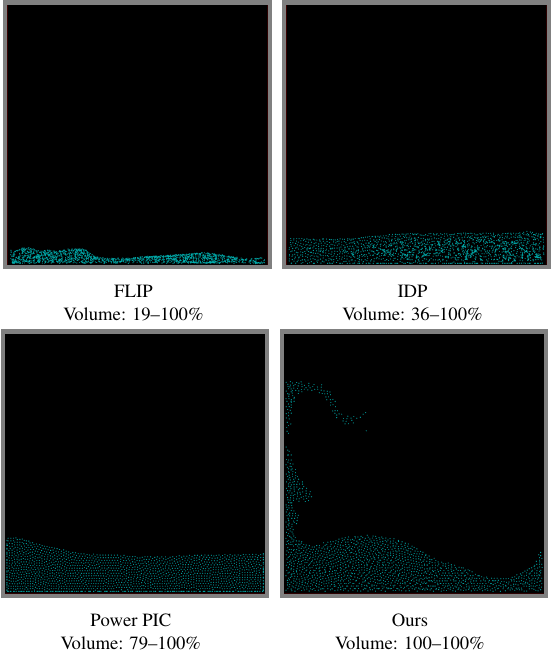}
    \caption{
    A dam scene with 1 ppc.
    }
    \label{fig.dam_2d_1ppc}
\end{figure}
\begin{figure}
    \centering
    \includegraphics[width=.8\columnwidth]{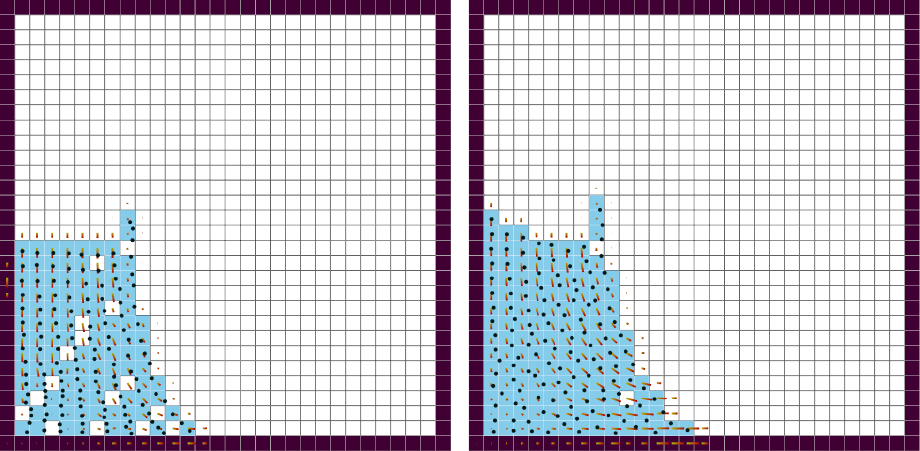}
    \caption{
    A dam scene on a $30 \times 30$ grid with 1 ppc.
    (Left) FLIP's fluid is riddled with holes that disrupt the velocity field. (Right) our fluid has less holes, and the constraint maintains the volume.
    }
    \label{fig.1ppc}
\end{figure}
\Cref{fig.dam_2d_1ppc} shows another 2D dam scene (without an emitter) using 1 ppc in a $\times 4$-finer grid (i.e., scaled by two in each dimension).
For FLIP and IDP, the fluid collapses, and there is a dramatic volume loss.
This is due to the sparse particle distribution (1 ppc), where particles can easily clump together, and some cells are missed.
As a result, the fluid is riddled with holes (see \cref{fig.1ppc}). The holes have zero pressure, and they disrupt the velocity field and attract particles.
Power PIC works hard to maintain a uniform distribution of the particles, and it requires significantly more Sinkhorn iterations for a time step. While its volume loss is less severe, the general behavior of the fluid is similar to FLIP.
Our method is the only one to maintain reasonable fluid volume and behavior, which is similar to the 4 ppc case, and the rare occurrences of holes in the fluid do not disrupt the velocity field.
Since there is only one particle at most in a grid cell, there can be no air bubbles, and the volume is perfectly preserved.

\begin{figure}
    \centering
    \includegraphics[width=1\columnwidth]{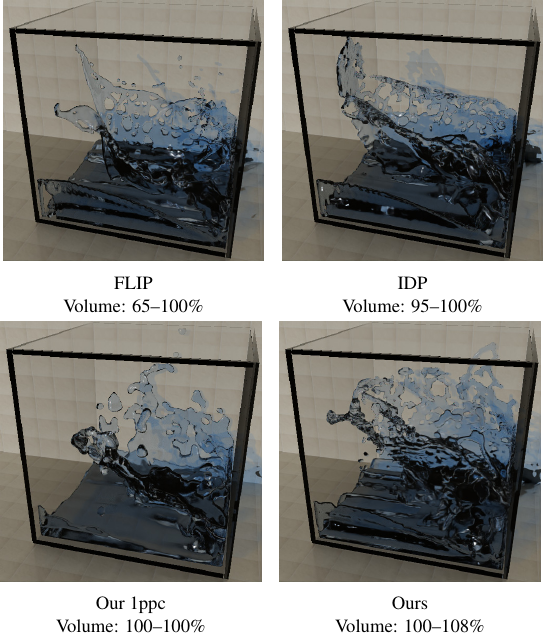}
    \caption{
    A dam scene.
    }
    \label{fig.dam_3d}
\end{figure}
\Cref{fig.dam_3d} shows a 3D dam scene.
FLIP loses a significant amount of volume.
IDP preserves the volume but suffers some compression.
IDP keeps the fluid smooth while ours introduces noise similar to the 2D case (\cref{fig.dam_noise}).
Using our method, the fluid hits the left wall the same time FLIP does. IDP overshoots the splash, which hits the wall earlier and more strongly.
The 1 ppc version of our method requires surface extraction of lower resolution, which is less detailed.
Our 8 ppc version gains some volume due to air bubbles.
\\

\begin{figure}
    \centering
    \includegraphics[width=1\columnwidth]{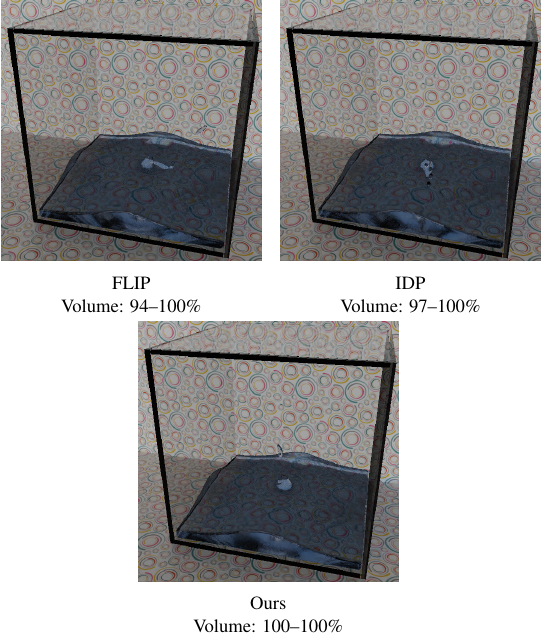}
    \caption{
    A water drop.
    }
    \label{fig.drop}
\end{figure}
\descrip{A drop of water.}
A drop of water is falling into a pool; see \cref{fig.drop}.
Notice where the splash goes. IDP’s throws the splash off the center while ours keeps it centered like FLIP.
\\

\begin{figure}
    \centering
    \includegraphics[width=.8\columnwidth]{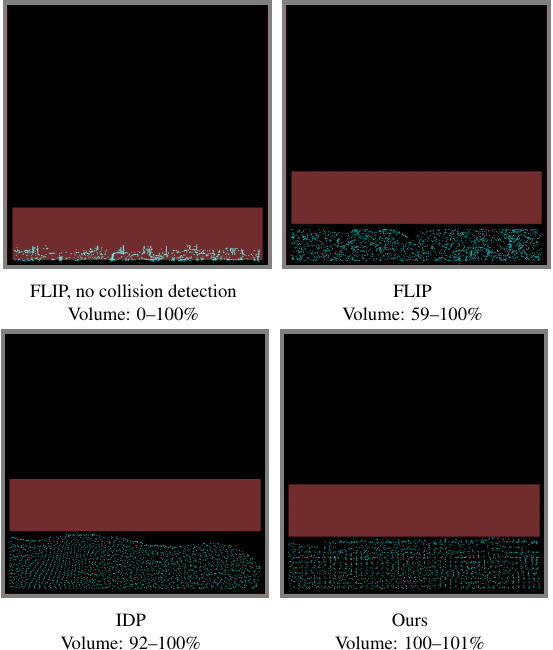}
    \caption{
    Compressing the fluid, the final frame.
    }
    \label{fig.compress}
\end{figure}
\descrip{Compressing the fluid.}
A heavy obstacle moves at a constant velocity towards the bottom of the tank. Its movement should supersede the fluid's unless fluid speed or incompressibility are compromised.
The obstacle's width is the same as the tank's, leaving no room for particles to pass it.
The expected result is the obstacle moving smoothly without overlapping any particles, compressing the fluid as much as it can; see \cref{fig.compress}.

We show two options for FLIP and IDP:
\begin{enumerate}
\item Moving the obstacle while disregarding the fluid. Since the fluid has no room to escape, there is an inevitable overlap with the obstacle, which leads to volume loss.

\item A naive collision detection, where the obstacle stops and waits until the fluid clears the cells that the obstacle is moving into.
IDP's volume correction disperses particles, which end up in the obstacle's way and obstruct its path more than FLIP.
Due to the jumpy behavior of the particles, IDP does not squeeze the fluid to the maximum possible, leaving some room for air.
FLIP, on the other hand, lets the obstacle compress the fluid too much, which leads to significant volume loss.
For both methods, the obstacle exhibits an undesired halting behavior.
\end{enumerate}
Out method achieves the desired behavior.
\\

\begin{figure}
    \centering
    \includegraphics[width=1\columnwidth]{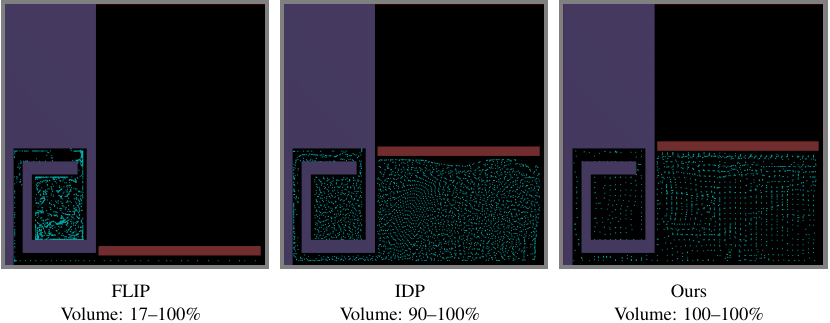}
    \caption{
    A spiral, the final frame.
    }
    \label{fig.spiral}
\end{figure}
\descrip{A spiral.}
The fluid is squeezed through a narrow spiral; see \cref{fig.spiral}.
Since the fluid’s speed is limited, so is the obstacle’s.
We used naive collision detection for FLIP and IDP.
FLIP lets the obstacle compress the fluid too much.
IDP allows the obstacle to lower more than it should before it can correct the fluid, consequently losing some volume that cannot be recovered.
In both methods, the obstacle's progress has more delays than necessary due to particles blocking the way.
Using our method, the obstacle progresses as fast as the fluid's speed limit allows, and the fluid is compressed as much as the volume restriction allows.
\\

\begin{figure}
    \centering
    \includegraphics[width=.8\columnwidth]{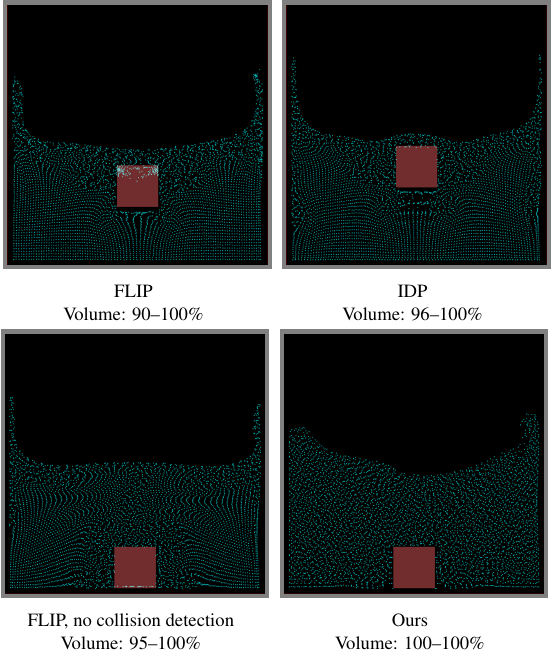}
    \caption{
    A falling obstacle, the final frame.
    }
    \label{fig.falling_obs_2d}
\end{figure}
\begin{figure}
    \centering
    \includegraphics[width=1\columnwidth]{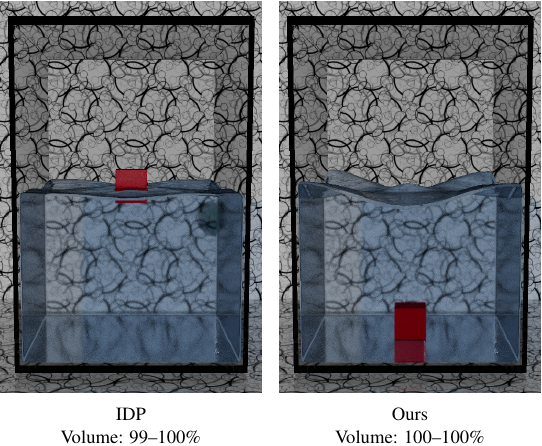}
    \caption{
    A falling obstacle, the final frame. IDP with naive collision detection.
    }
    \label{fig.falling_obs_3d}
\end{figure}
\descrip{A falling obstacle.}
An obstacle is falling into the water.
FLIP and IDP behave similarly:
\begin{itemize}
\item Without collision detection, some particles are trapped at the bottom, which leads to volume loss.
\item With collision detection, the obstacle movement is halted not long after hitting the water, far from the bottom of the tank, and there is no progress.
\end{itemize}
Using our method, the obstacle moves smoothly (like FLIP without collision detection) and there is no overlap with particles (which causes volume loss).
\Cref{fig.falling_obs_2d} shows the 2D case, and \cref{fig.falling_obs_3d} shows the 3D case.

\begin{figure}
    \centering
    \includegraphics[width=.7\columnwidth]{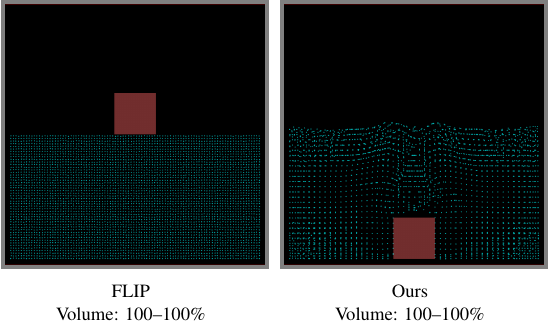}
    \caption{
    A falling obstacle with boundary conditions set to zero velocity, the final frame.
    }
    \label{fig.falling_obs_2d_zero_BC}
\end{figure}
Our method can move the obstacle through the fluid even without the obstacle exerting any forces on the fluid.
\Cref{fig.falling_obs_2d_zero_BC} shows an experiment where the boundary conditions for the pressure equation along the obstacle's boundary, \cref{eq.solid_boundary}, are set to zero velocity.
There is nothing to repel the fluid from the obstacle’s way.
As expected, when FLIP uses collision detection, the obstacle cannot penetrate the fluid.
Using our method, the obstacle progresses smoothly through the fluid, where the correction method displaces particles out of the obstacle's way, requiring no other forces.

\begin{figure}
    \centering
    \includegraphics[width=1\columnwidth]{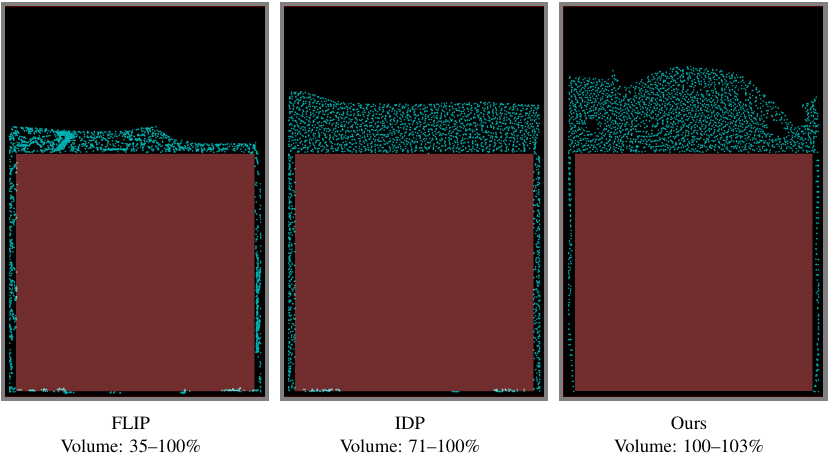}
    \caption{
    A large falling obstacle, the final frame.
    }
    \label{fig.large_falling_obs_2d}
\end{figure}
\Cref{fig.large_falling_obs_2d} and \cref{fig.large_falling_obs_3d} show another variation with a large obstacle falling into the water.
The obstacle's speed should depend on how fast the fluid can flow along the narrow paths between the obstacle and the tank.
Using collision detection for FLIP and IDP, the object makes no progress after hitting the surface. Without collision detection, particles on the bottom of the tank are trapped inside the obstacle, leading to significant volume loss.

\subsection{Discussion}
\begin{table}
    \centering
    \includegraphics[width=.8\columnwidth]{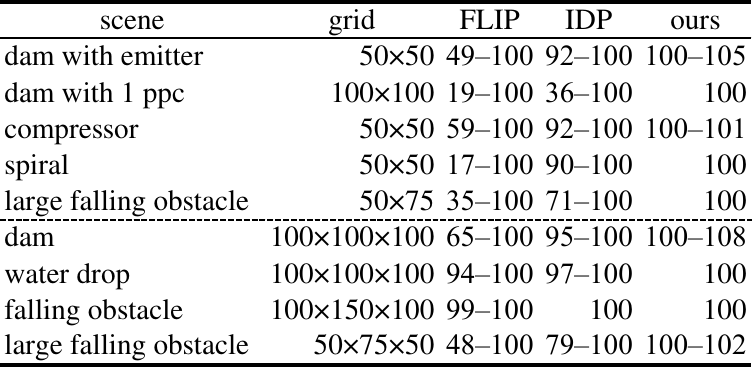}
    \caption{Volume preservation.
    A method's column shows the range of the fluid volume (presented as percentage of how much it should occupy) over all the simulation iterations, based on \cref{eq.volume}.
    }
    \label{tab.vol}
\end{table}
\begin{table}
    \centering
    \includegraphics[width=1\columnwidth]{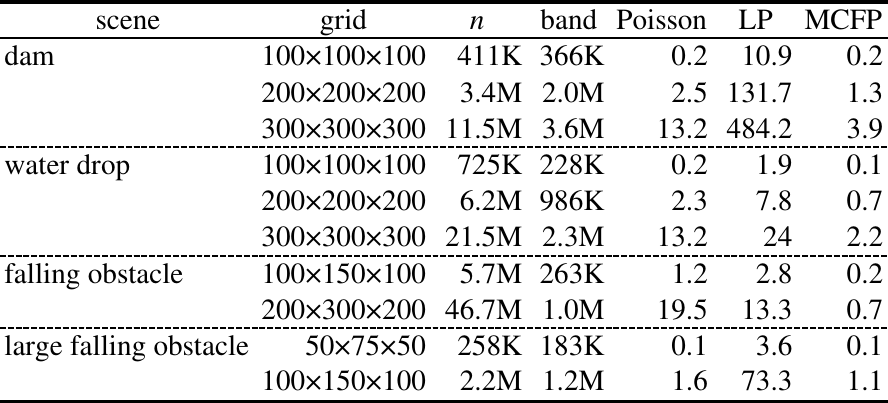}
    \caption{Running time.
    "$n$": number of particles in the scene.
    "band": the average number of particles in the band.
    "Poisson", "LP", and "MCFP": the average time it takes to solve a Poisson equation, the LP problem in \cref{eq.problem3}, and the MCFP problem in \cref{sec.push_part}.
    Average quantities are calculated over all iterations.
    Timings are given in seconds, rounded to one decimal place.
    }
    \label{tab.time}
\end{table}
\descrip{Methods.}
We focused the experiments on comparison with IDP, where its paper shows comparison with several other methods.
We did not use the band method for FLIP to keep the settings close to IDP, which does not support it. Also, both methods did not have performance issues that would require it.
\\

\descrip{Behavior and volume preservation.}
IDP allows the fluid to violate incompressibility and increase density. In its correction step, IDP moves particles to improve density accuracy.
The improvement is gradual, and the fluid may already be in a state that it cannot be recovered from.
We offered several scenarios to challenge this aspect, offering two reasonable solutions to address collision detection for FLIP and IDP: ignoring the particles and a naive detection approach. Even if the user manually selects the best of the two for each scenario, none of the behaviors are quite acceptable. The naive method caused a halting behavior or even a premature (complete) stop. Ignoring particles led to an inevitable overlap between the obstacle and the particles, which caused volume loss. Even if the loss was acceptable, the progress of the obstacle was smooth and arbitrary instead of being dependent on the fluid speed (e.g., the spiral scene).
In contrast, our method strictly enforces incompressibility.
A full correction is applied immediately, and the fluid cannot be compressed.
Furthermore, while the obstacle's movement is prioritized, its speed is still limited by how fast the fluid can clear the way.

Our correction is mostly done by blocking and preventing particles from moving into neighboring cells rather than push them around.
If a particle is moved to another cell, then the particle is positioned within the cell to be as close as possible to the location it was supposed to go.
This is also prioritized over a more uniform particle distribution, like in Power PIC, which looks nice in 2D but affects the fluid behavior.
IDP tends to preserve FLIP's behavior, e.g., \cref{fig.dam_noise}, where it keeps a smooth surface and clumped particles vs the other methods that add noise.
However, the correction movements can also have substantial influence on the fluid's behavior, e.g., in the 3D dam when the splash overshoots and hits the left wall (\cref{fig.dam_3d}), or in the water drop scene, where the splash of water is thrown off the center  (\cref{fig.drop}).
\\

\descrip{Running time.}
In 2D, performance was not an issue for any of the methods, and a time step took less than a second.
\Cref{tab.time} gives timings for 3D scenes with varying grid sizes.

The timing of solving a Poisson equation depends on the number of fluid cells.
For methods that preserve incompressibility, the number of fluid cells is approximately the number of particles divided by $\mu$.

\begin{figure}
    \centering
    \includegraphics[width=.4\columnwidth]{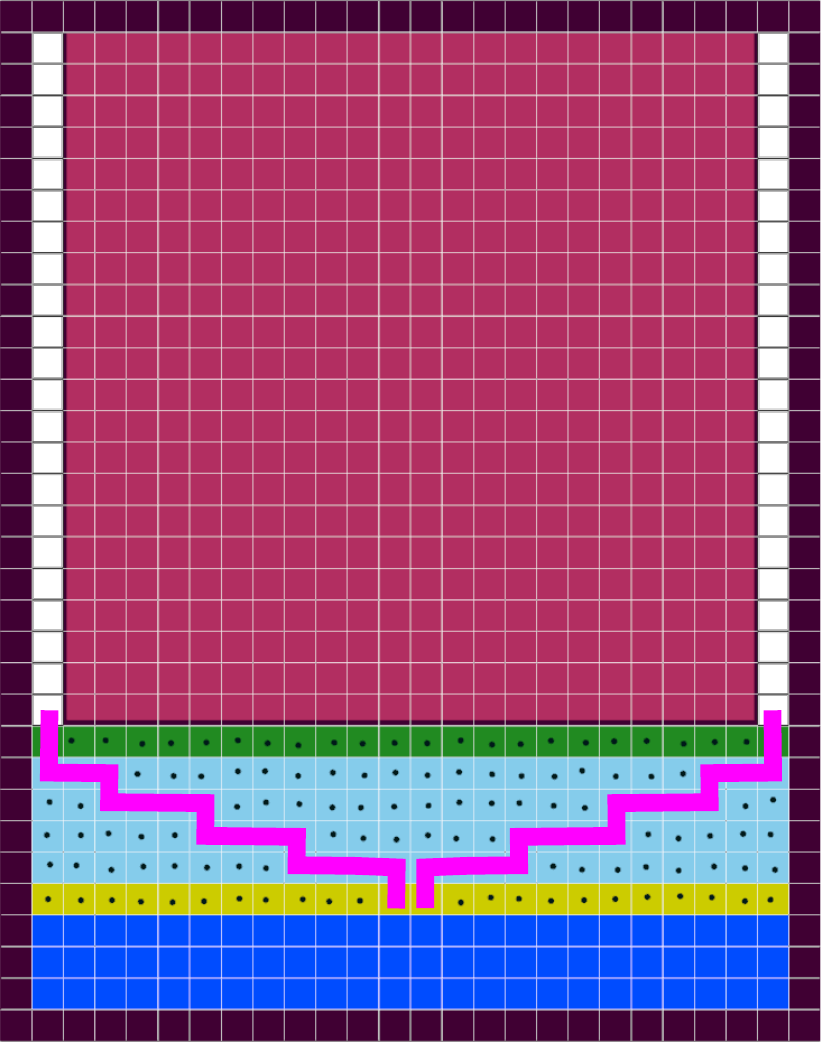}
    \caption{
    An extreme (hypothetical) case of paths (in magenta) that start in the middle of the band, creating a V shape and blocking other paths.
    }
    \label{fig.v_paths}
\end{figure}
MCFP's timing depends on the number of fluid cells.
In a typical scene, \cref{alg.find} needs to be executed rarely more than once in a time step if at all.
However, in a scene such as the large falling object, there are time steps where a few calls are needed.
For example, consider the extreme scenario in \cref{fig.v_paths}.
Two low-cost paths start in two incident cells in the middle of the band interface and lead to the surface along the narrow passages between the obstacle and the tank, creating a V shape. Since the algorithm finds non-overlapping paths, these two paths block all other paths, and another call to \cref{alg.find} is necessary if more paths are needed.
This increases the MCFP's average time for such a scene.

The time required to solve the LP can vary, depending on how hard the problem is, which does not necessarily depend on the number of particles and grid size.
For example, in the first iterations of the water drop scene, before the drop hits the pool, the LP solution is close to the ideal particle positions, and the LP is solved in 7 seconds for a $300^3$ grid.
On the other hand, in the large obstacle scene, the objective of the LP includes the obstacle, and the solution decides if the obstacle moves or not. This creates a dependency between cells in $\mathcal{C}_\text{new\_solid}$ that are fluid or incident to fluid cells, which is similar to a global dense constraint. A solution to the LP in this case can take a few minutes even for a $50 \times 75 \times 50$ grid.

An obvious advantage of using the band method is the reduction in the number of particles, e.g., from 46.7 million to an average of one million in the falling obstacle scene with the $200 \times 300 \times 200$ grid.
But another advantage of the band method is that it may accelerate the LP even if the number of particles is not reduced significantly. For example, in the $100^3$ dam scene, when not using the band method, the average LP time is 38.9 seconds. When using the band method, the number of particles is reduced only by 16\%, but the average LP time is reduced to 8.6 seconds. This is because the band method simplifies the problem. Intuitively, the constraints become local, where particles can simply follow gravity into the deep with no restriction (which is corrected by the MCFP in the second step).
This is not the case without the band method, where the flow reaches the bottom of the tank or abides the volume constraint, and it needs to go side-ways and up, which is a more global behavior.

Comparing the three methods, FLIP is the fastest since it requires only one Poisson solution, which the two other methods require as well.
IDP requires solving an additional Poisson equation while our method needs to solve an LP and an MCFP.
The cost of the additional Poisson solution is reasonable, and more importantly it is predictable since it depends on the number of fluid cells.
Our method mostly runs in a reasonable time on moderate size grids.
But in some cases, while solving the MCFP remains reasonably fast, solving the LP can take a few minutes, which leaves room for improvement.

\section{Conclusion}
We proposed a method that constrains particles to grid cells to enforce our definition of discrete incompressibility. While the fluid can still inflate with air bubbles, we show experimentally that the expansion is moderate.
Keeping strict incompressibility is one advantage over previous work, which instead gradually corrects the fluid over time. One issue with gradual correction is that volume preservation is not perfect, which may cause noticeable artifacts. A more sever issue is that the fluid can reach a state that is irrecoverable.

Our framework can be further exploited in other applications, and we show examples of coupling with solids which naive solutions applied to the state of the art fail to handle adequately.

The main drawback of the method is performance. In each iteration, an LP is solved. Besides the number of particles, the fluid configuration affects the running time, which may be longer than desired. We offered acceleration via an adapted version of the band method that enforces incompressibility, and we showed experimentally that it performs reasonably on moderate size grids.
The fastest variation of our band method solves an easier LP, followed by an additional correction that solves an MCFP. While the solution is not optimal, the result is reasonable for the affected amount of particles.
A future avenue could be to find faster alternatives to the LP.
\printbibliography

\appendix 

\section{Proofs} \label{sec.proofs}
\begin{figure}
    \centering
    \includegraphics[width=.8\columnwidth]{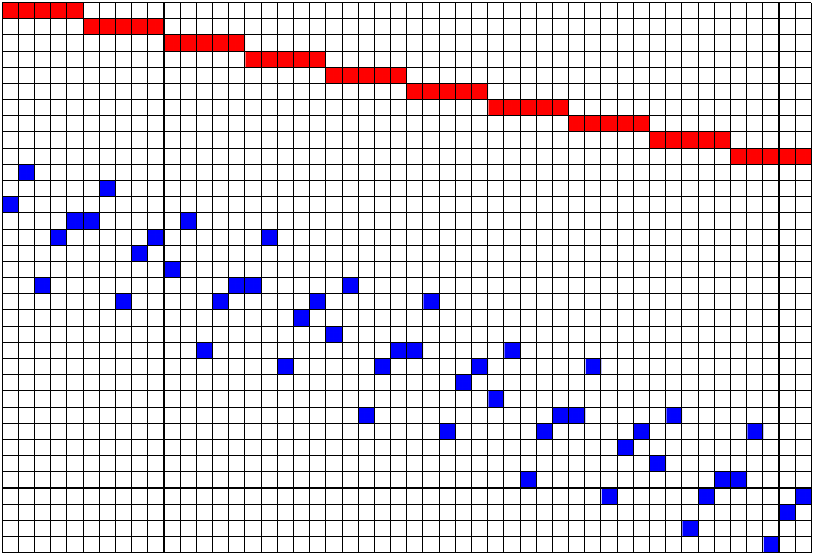}
    \caption{
    Visualizing a TU matrix.
    Zeros are in white, the rest are ones.
    The rows with the red cells correspond to \cref{eq.problem3.c}.
    The rows with the blue cells correspond to \crefrange{eq.problem3.d}{eq.problem3.f} (without duplicate rows).
    In both the red and blue set of rows, each column sums up to one.
    }
    \label{fig.draw_TU}
\end{figure}
\LPrelax*
\begin{proof}
Transform the LP in \cref{eq.problem3} into a canonical form $\max \left\{ cy \mid Ay \le d \right\}$, where $A$ is a matrix, and $c$, $y$, and $d$ are vectors:
\begin{itemize}
\item Change the objective to $\max$ by negating it.

\item Replace an equality constraint with two inequalities (bounding the LHS expression from both sides).

\item Change $\ge$ inequalities to $\le$ by negating them.

\item Convert the problem into a matrix form.
An expression that is bounded from both sides (which appears in two inequalities, e.g., \cref{eq.problem3.e} or a transformed equality constraint) appears as two identical rows in $A$ up to a sign.
\end{itemize}
The feasible region of an LP is a polyhedron (an intersection of hyperplanes).
Due to linearity, an optimal solution (an extreme point) is at a polyhedron vertex.
The polyhedron has vertices with integral coordinates if the matrix $A$ is totally unimodular (TU) and $d$ (the RHS) is integral \cite[theorem 19.1]{schrijver98theory}.
A matrix is TU if each of its subdeterminants is $\in \left\{ 0, \pm 1 \right\}$.
\begin{lemma}
$A$ is TU.
\end{lemma}
Proof by induction on the size $k \times k$ of a square submatrix of $A$.

Base case: holds for $k=1$ since each entry of $A$ is $\in \left\{ 0, \pm 1 \right\}$.

Induction step: Assume the determinant of a $k \times k$ submatrix of $A$ is $\in \left\{ 0, \pm 1 \right\}$, prove for a submatrix $B \in \R^{(k+1) \times (k+1)}$.
Possible cases:
\begin{itemize}
\item $B$ has a row or column of zeros. Then, it is rank-deficient, and its determinant is zero.
Similarly if $B$ has a duplicate row up to a sign (e.g., two inequalities that bound the same expression, and both rows are in $B$).

\item $B$ has a row with a single nonzero $B_{ij}$ (e.g., \cref{eq.problem3.b}). Then, consider the Laplace expansion along this row. It will be equal to $B_{ij}$ times a $k \times k$ cofactor of $B$, which according to the assumption is $\in \left\{ 0, \pm 1 \right\}$.
Similarly if $B$ has a column with a single nonzero

\item Each variable corresponds to a particle movement, which ends up in a specific cell. Therefore, each variable appears only once in \crefrange{eq.problem3.d}{eq.problem3.f}.
Moreover, each variable appears once in \cref{eq.problem3.c}.
This leaves us with the last case where each column of $B$ has two nonzeros.
The nonzeros in each row are either all 1 or -1. Multiply each negative row by -1, which may only affect the sign of the determinant.
Divide $B$ into two matrices
\begin{equation*}
\begin{bmatrix}
B_1 \\
B_2
\end{bmatrix}
\ , \quad B_1 \in \R^{l \times (k+1)}
\ , \  B_2 \in \R^{(k-l+1) \times (k+1)}
\end{equation*}
such that a column in each matrix has a single 1; see \cref{fig.draw_TU} for illustration.
Let $v \in \R^k$ be the vector
\begin{equation*}
v \coloneq
\begin{bmatrix}
v_1 \\
v_2
\end{bmatrix}
\ , \quad v_1 \coloneq \mathbb{1} \in \R^l
\ , \  v_2 \coloneq -\mathbb{1} \in \R^{k-l+1}
\ ,
\end{equation*}
where $\mathbb{1}$ is a vector of ones.
$v$ is in the null space of $B^\intercal$, and thus $\det B = 0$.
\end{itemize}
\end{proof}

\end{document}